\documentclass[pra,aps]{revtex4}
\pdfoutput=1
\usepackage[utf8]{inputenc}
\usepackage[mathletters]{ucs}
\usepackage[dvips]{graphicx}
\usepackage{amsmath,amsthm,amssymb}
\usepackage{youngtab}

\newcommand{\<}{\langle}
\renewcommand{\>}{\rangle}
\newcommand\be{\begin{equation}}
\newcommand\ee{\end{equation}}
\newcommand\ot{\otimes}
\newcommand\C{\mathbb{C}}
\newcommand\cH{\mathcal{H}}
\newcommand\sfS{{\mathsf S}}
\newcommand\sfA{{\mathsf A}}
\newcommand\sfX{{\mathsf X}}

\newcommand\hcal{\cH}
\newcommand\bea{\begin{array}}
\newcommand\eea{\end{array}}
\newcommand\ben{\begin{eqnarray}}
\newcommand\een{\end{eqnarray}}

\newcommand\bei{\begin{itemize}}
\newcommand\eei{\end{itemize}}
\newcommand\bee{\begin{enumerate}}
\newcommand\eee{\end{enumerate}}

\usepackage{marginnote}

\def\ot{\otimes}

\def\bei{\begin{itemize}}
\def\eei{\end{itemize}}

\newtheorem{lemma}{Lemma}
\newtheorem{corollary}{Corollary}
\newtheorem{theorem}{Theorem}
\newtheorem{example}{Example}
\newtheorem{proposition}{Proposition}
\newtheorem{definition}{Definition}

{\par\addvspace{\medskipamount}\noindent\textbf{Examples.}\hspace{1ex}}%
{\par\medskip}

\newenvironment{remark}%
{\par\addvspace{\medskipamount}\noindent\textbf{Remark.}\hspace{1ex}}%
{\par\medskip}

\DeclareUnicodeCharacter{981}{\phi}
\DeclareUnicodeCharacter{348}{{\hat{S}}}
\DeclareUnicodeCharacter{8224}{{\dagger}}

\usepackage{color}

\begin{document}
\reversemarginpar

\title{Distillation of entanglement by projection on permutationally invariant subspaces}

\author{Miko{\l}aj Czechlewski$^1$, Andrzej Grudka$^{1,3}$, Micha{\l} Horodecki$^{2,3}$, Marek Mozrzymas$^4$ and Micha{\l} Studzi{\'n}ski$^{2,3}$}

\affiliation{$^1$Faculty of Physics, Adam Mickiewicz University, 61-614 Pozna\'{n}, Poland\\
$^2$Institute for Theoretical Physics and Astrophysics,
University of Gda{\'n}sk, 80-952 Gda{\'n}sk, Poland\\
$^3$National Quantum Information Centre of Gda\'{n}sk, 81-824 Sopot, Poland\\
$^4$Institute for Theoretical Physics,
University of Wroc{\l}aw, 50-204 Wroc{\l}aw, Poland}

\date{\today}


\begin{abstract}
We consider distillation of entanglement from two qubit states which are mixtures of three mutually orthogonal states: two pure entangled states and one pure product state. We distill entanglement from such states by projecting $n$ copies of the state on permutationally invariant subspace and then applying one-way hashing protocol. We find analytical expressions for the rate of the protocol. We also generalize this method to higher dimensional systems. To get analytical expression for two qubit case, we faced a mathematical problem of diagonalizing a family of matrices enjoying some symmetries w.r.t. to symmetric group. We have solved this problem in two ways: (i) directly, by use of Schur-Weyl decomposition and Young symmetrizers (ii) showing  that the problem is equivalent to a problem of diagonalizing adjacency matrices in 
a particular instance of a so called algebraic association scheme.
\end{abstract}

\maketitle 

\section{Introduction}
Pure entanglement is fundamental resource in quantum information \cite{Bennett5,Bennett6,Ekert1}. However, usually the parties who want to perform some communication task have access to mixed entanglement. In such a case in order to obtain useful entanglement they should be able to distill pure entanglement -- usually in the form of maximally entangled pairs. Procedures which allow to distill entanglement are called distillation protocols and are realized by means of local operations and classical communication \cite{Bennett1,Bennett3,PhysRevLett.77.2818,Dur}. Let us suppose that two parties -- Alice and Bob -- share $n$ copies of a state $\rho$, they process them by a protocol $P$ and obtain $m$ copies of maximally entangled pairs. The ratio $\lim_{n\rightarrow\infty}\frac{m}{n}$ is called the rate of the protocol $P$ with respect to state $\rho$. The maximum of the rate over all distillation protocols is called distillable entanglement of a state $\rho$. Distillable entanglement is difficult to calculate and is only known for certain states, i.e., bound entangled states or maximally correlated state \cite{Horodecki7, Rains1, Eisert1, Chen1, Hamieh1, Hiroshima1}. However, one can always find a lower bound on distillable entanglement by calculating rate of a particular protocol.

In \cite{PhysRevA.80.014303} authors introduced an efficient protocol for two qubit states which are mixtures of one pure entangled state and one pure product state which are orthogonal to each other. In the present manuscript we apply this protocol to two qubit states which are mixtures of three mutually orthogonal states: two pure entangled states and one pure product state. We find analytical expressions for the rate of the protocol. Moreover, we generalize the protocol to entangled state of qudits, i.e., $d$-dimensional quantum system.

To obtain the analytical expression for two-qubit case, we face a problem of diagonalizing a 
family of matrices, which arise from projecting $n$ copies of a state
diagonal in a $\frac{1}{\sqrt2}(|0\>\pm |1\>)$ basis onto a subspaces spanned by vectors 
of fixed number of $1$'s in computational basis. We present two solutions to the problem.
The first solution exploits group theoretical methods such as Schur-Weyl decomposition and Young symmetrizers.
The second method refers to so called {\it algebraic association schemes} \cite{Bannai}. It turns out 
that our problem is directly related to diagonalization of so called {\it adjacency} matrices in a 
particular algebraic association acheme called {\it Johnson scheme}, whose solution is known.

The paper is organized as follows. In Section II we describe the basic protocol for entangled states of two qubits. The protocol consists of two parts: measurement of $n$ copies of the state and application of one way hashing protocol to the post-measurement state. In Section III we generalize this protocol to higher dimensional systems. In Section IV we calculate coherent information of the post-measurement state of Section II, i.e. the rate of one-way hashing protocol. The main effort is here to find analytically the eigenvalues 
of a family of matrices.  In Sec. \ref{subsec:proof} we obtain the form of eigenvalues via two different methods:  in Sec. \ref{subsec:eig-A-Young} via group theoretic methods, and in Sec. \ref{subsubsec:eig-A-MM} via algebraic association schemes. In Section V we present rates of the protocol for various states of Section II.

\section{Basic protocol for entangled states of two qubits}

Let Alice and Bob share $N=2^k$ copies of a state
\begin{eqnarray}
& \rho_{AB}=x(q|\Phi^{+}(\alpha)\rangle\langle\Phi^{+}(\alpha)|_{AB}+(1-q)|\Phi^{-}(\alpha)\rangle\langle\Phi^{-}(\alpha)|_{AB})+\nonumber\\
& +(1-x)|01\rangle\langle01|_{AB},
\label{MA:1}
\end{eqnarray}
where
\begin{eqnarray}
|\Phi^{\pm}(\alpha)\rangle_{AB}=\sqrt{\alpha}|00\rangle_{AB}\pm \sqrt{1-\alpha}|11\rangle_{AB}.
\end{eqnarray}
First Alice and Bob project their parts of the state on subspace spanned by vectors with definite number of $0$'s and $1$'s. If Alice finds the same numbers of $0$'s and $1$'s as Bob then they perform one-way hashing protocol. If Alice finds different numbers of $0$'s and $1$'s than Bob then they divide $N$ pairs of qubits into two equal groups and perform analogous measurements on each group independently. 
The probability that Alice and Bob succeed in the first step, i.e., Alice finds the same numbers of $0$'s and $1$'s as Bob, is equal to the probability of having $2^k$ states $q|\Phi^{+}(\alpha)\rangle\langle\Phi^{+}(\alpha)|_{AB}+(1-q)|\Phi^{-}(\alpha)\rangle\langle\Phi^{-}(\alpha)|_{AB}$ because terms containing states $|01\rangle\langle01|$ are not in the subspace on which Alice and Bob project the state, i.e, it is equal to $p(S_1)=x^{2^k}$. We assume that if they succeed in the first step they can distill entanglement from the post-measurement state at partial rate $R_1$. If in the first step Alice and Bob do not succeed, then in the second step Alice and Bob can succeed at most for one group of pairs of qubits. The probability that Alice and Bob succeed for one group of pairs of qubits in the second step and they do not succeed in the first step is equal to $p(S_2,F_1)=2p(s_2)p(f_2)$, where $p(s_2)=x^{2^{k-1}}$ is probability of having $2^{k-1}$ states  $q|\Phi^{+}(\alpha)\rangle\langle\Phi^{+}(\alpha)|_{AB}+(1-q)|\Phi^{-}(\alpha)\rangle\langle\Phi^{-}(\alpha)|_{AB}$ in a group of $2^{k-1}$ pairs of qubits, $p(f_2)=1-x^{2^{k-1}}$ is probablity of not having $2^{k-1}$ states  $q|\Phi^{+}(\alpha)\rangle\langle\Phi^{+}(\alpha)|_{AB}+(1-q)|\Phi^{-}(\alpha)\rangle\langle\Phi^{-}(\alpha)|_{AB}$ in a group of $2^{k-1}$ pairs of qubits. The factor $2$ stands because Alice and Bob can succeed for the first or the second group of pairs of qubits. We assume that if they succeed in the second step they can distill entanglement from the post-measurement state at partial rate $R_2$. Moreover, Alice and Bob divide a group of $\frac{N}{2}$ pairs of qubits for which they did not succeed into two equal groups and perform analogous measurements on each group independently. They repeat the procedure until $k-1$-th step (there is no sense to perform the measurement on one pair of qubits).  In general the probability that Alice and Bob succeed in the $i$-th step for one of two groups of $2^{k-i+1}$ pairs of qubits and they did not succeed in the $i-1$-th step for a group of $2^{k-i+2}$ pairs of qubits (they also did not succeed in all previous steps for any group of qubits containing the latter group) is equal to $p(S_i,F_{i-1})=2p(s_i)p(f_i)$, where $p(s_i)=x^{2^{k-i+1}}$ is probability of having $2^{k-i+1}$ states  $q|\Phi^{+}(\alpha)\rangle\langle\Phi^{+}(\alpha)|_{AB}+(1-q)|\Phi^{-}(\alpha)\rangle\langle\Phi^{-}(\alpha)|_{AB}$ in a group of $2^{k-i+1}$ pairs of qubits, $p(f_i)=1-x^{2^{k-i+1}}$ is probablity of not having $2^{k-i+1}$ states  $q|\Phi^{+}(\alpha)\rangle\langle\Phi^{+}(\alpha)|_{AB}+(1-q)|\Phi^{-}(\alpha)\rangle\langle\Phi^{-}(\alpha)|_{AB}$ in a group of $2^{k-i+1}$ pairs of qubits. 
Hence the total rate of the protocol is 
\begin{eqnarray}
& R =\frac{1}{2^k}(p(S_1)R_1+p(S_2,F_{1})R_2+...+\nonumber\\
& +2^{i-2}p(S_i,F_{i-1})R_i+...)=\nonumber\\
& =\frac{1}{2^k}(p(s_1)R_1+2p(s_2)p(f_2)R_2+...+\nonumber\\
& +2^{i-2}2p(s_i)p(f_i)R_i+...)=\nonumber\\
& =\frac{1}{2^k}(x^{2^k}R_1+2(1-x^{2^{k-1}})x^{2^{k-1}}R_2+...+\nonumber\\
& +2^{i-2}2(1-x^{2^{k-i+1}})x^{2^{k-i+1}}R_i+...)=\nonumber\\
& =\frac{1}{2^k}(x^{2^k}R_1+\sum_{i=2}^{k-1}2^{i-1}(1-x^{2^{k-i+1}})x^{2^{k-i+1}}R_i).
\end{eqnarray}
The factor $\frac{1}{2^k}$ stands because Alice and Bob start with $N=2^k$ copies of a state $\rho_{AB}$ and the factor $2^{i-2}$ stands because in the $i-1$-th step Alice and Bob could have $2^{i-2}$ groups of pairs of qubits for which they did not succeed.
It is convenient to write the total rate of the protocol in the following form
\begin{eqnarray}
& R=\frac{1}{2^k}\sum_{i=1}^{k-1}x^{2^{k-i+1}}(2^{i-1}R_{i}-2^{i}R_{i+1}),
\label{MA:4}
\end{eqnarray}
with $R_k=0$.

Let us now calculate partial rates $R_i$. The probability that Alice finds $l$ $0$'s and $2^{k-i+1}-l$ $1$'s in a group of $n=2^{k-i+1}$ pairs of qubits under the condition that Alice finds the same numbers of $0$'s and $1$'s as Bob is equal to
\begin{equation}
p(l|S_i)=\alpha^l(1-\alpha)^{(2^{k-i+1}-l)}{2^{k-i+1} \choose l}
\end{equation}
and the post-measurement state is 

\be
\rho_{lAB}^{(n)}=\frac{P_{lA}\otimes P_{lB}\rho_{AB}^{\ot n} P_{lA}\otimes P_{lB}}{
\text{Tr}(P_{lA}\otimes P_{lB}\rho_{AB}^{\ot n} P_{lA}\otimes P_{lB})}
\ee
where $P_l$ are projectors which project onto a subspace of
$(\C^2)^{\ot n}$ spanned by all standard basis vectors, having $l$ $1$'s and $n-l$ $0$'s,
such as $|\underbrace{0\ldots 0}_{n-l}\underbrace{1\ldots 1}_l\>$. Note that the post-measurement state does not depend on results of previous measurements \cite{PhysRevA.80.014303}.
Having large number of copies of a state $\rho^{(n)}_{lAB}$ they can apply one-way hashing protocol and distill entanglement at rate equal to coherent information $I_c$ of state $\rho^{(n)}_{lAB}$, where $I_c=S(\rho^{(n)}_{lB})-S(\rho^{(n)}_{lAB})$ \cite{Devetak1,Devetak2}. Hence the partial rates are 
\begin{equation}
R_i=\sum_{l=0}^{2^{k-i+1}}\alpha^l(1-\alpha)^{(2^{k-i+1}-l)}{2^{k-i+1} \choose l}(S(\rho^{(2^{k-i+1})}_{lB})-S(\rho^{(2^{k-i+1})}_{lAB})).
\end{equation}

\section{Generalized protocols for entangled states of qudits}

Let us consider the following state
\begin{eqnarray}
\rho_{AB}=x|\Phi_d^+\rangle\langle\Phi_d^+|_{AB}+(1-x)|01\rangle\langle01|_{AB},
\label{MA:8}
\end{eqnarray}
where
\begin{eqnarray}
|\Phi_d^+\rangle_{AB}=\frac{1}{\sqrt{d}}\sum_{i=0}^{d-1}|ii\rangle_{AB}.
\label{MA:9}
\end{eqnarray}
Let Alice and Bob apply a similar protocol as before, i.e., in successive steps they project their parts of $n=2^{k-i+1}$ copies of pairs of qudits on a subspace spanned by vectors with definite numbers of $0$'s, $1$'s, $2$'s and so on. If Alice and Bob measure the same numbers of $0$'s, $1$'s, $2$'s and so on, then the post-measurement state is maximally entangled state of Schmidt rank
\begin{eqnarray}
r_l=\frac{2^{k-i+1}!}{l_0!l_1!...l_{d-1}!}
\end{eqnarray}
where $l_0, l_1,... l_{d-1}$ are numbers of $0$'s, $1$'s, ...$d-1$'s found by each party. 
Hence the partial rates are given by
\begin{eqnarray}
& R_i=\frac{1}{d^{2^{k-i+1}}} \nonumber\\
& \sum_{l_0,l_1,...l_{d-1}}\frac{2^{k-i+1}!}{l_0!l_1!...l_{d-1}!}\log(\frac{2^{k-i+1}!}{l_0!l_1!...l_{d-1}!}).
\end{eqnarray}
where the sum runs over $0\leq l_0,l_1,...l_{d-1}\leq2^{k-i+1}$ satisfying the constraint $l_0+l_1+...l_{d-1}=2^{k-i+1}$.
However there exists a protocol which achieves higher rates. It happens that projective measurements performed by Alice and Bob are too invasive, i.e., they destroy too much entanglement.
Let us instead define projectors
\begin{eqnarray}
P_{lA(B)}=(P^{0}_{A(B)})^{\otimes l}(P^{\bar{0}}_{A(B)})^{\otimes n-l}+\text{permutations},
\end{eqnarray}
where
\begin{eqnarray}
& P^{0}_{A(B)}=|0\rangle\langle0|_{A(B)}\nonumber\\
& P^{\bar{0}}_{A(B)}=I_{A(B)}-|0\rangle\langle0|_{A(B)},
\end{eqnarray}
and let in successive steps both Alice and Bob perform mesurements given by these projectors on $n=2^{k-i+1}$ copies of pairs of qudits.
These projectors discriminate the number of $|\Phi_d^+\rangle_{AB}$ states versus the number of $|01\rangle_{AB}$ states as well as projectors which project on a subspace spanned by vectors with definite numbers of $0$'s, $1$'s, $2$'s and so on and they are less invasive. If Alice and Bob measure the same numbers of $0$'s, i.e., they both obtain $P_{lA(B)}$ as a result of the measurement, then the post-measurement state is maximally entangled state of Schmidt rank
\begin{eqnarray}
r_l={2^{k-i+1} \choose l} (d-1)^{(2^{k-i+1}-l)}.
\end{eqnarray}
Hence the partial rates are 
\begin{eqnarray}
& R_i=\frac{1}{d^{2^{k-i+1}}} \nonumber\\
& \sum_{l}{2^{k-i+1} \choose l} (d-1)^{(2^{k-i+1}-l)} \nonumber\\
& \log({2^{k-i+1} \choose l} (d-1)^{(2^{k-i+1}-l)}).
\label{MA:14}
\end{eqnarray}
Moreover the total rate of the protocol is given by Eq. \ref{MA:4} with the sum extended from $1$ to $k$ and $R_{k+1}=0$, because now Alice and Bob can distill entanglement by performing measurement even on a single copy.

As a further example let us consider the following state
\begin{eqnarray}
\rho_{AB}=x|\Phi_d^+\rangle\langle\Phi_d^+|_{AB}+(1-x)\sum_{i, \text{even}}\frac{2}{d}|ii+1\rangle\langle ii+1|_{AB}.
\label{MA:15}
\end{eqnarray}
where $|\Phi_d^+\rangle_{AB}$ is given by Eq. \ref{MA:9}.
Let us define projectors
\begin{eqnarray}
P_{lA(B)}=(P^{\text{even}}_{A(B)})^{\otimes l}(P^{\text{odd}}_{A(B)})^{\otimes (n-l)}+\text{permutations},
\end{eqnarray}
where
\begin{eqnarray}
& P^{\text{even}}_{A(B)}=\sum_{i,\text{even}}|i\rangle\langle i|_{A(B)}\nonumber\\
& P^{\text{odd}}_{A(B)}=\sum_{i,\text{odd}}|i\rangle\langle i|_{A(B)},
\end{eqnarray}
and let  both Alice and Bob perform measurements on $n=2^{k-i+1}$ copies of pairs of qudits.
If they measure the same numbers of even $i$'s, i.e., both Alice and Bob obtain $P_{lA(B)}$ as a result of the measurement, then the post-measurement state is maximally entangled state of Schmidt rank
\begin{eqnarray}
r_l={2^{k-i+1} \choose l} (d/2)^{l}(d/2)^{(2^{k-i+1}-l)}. 
\end{eqnarray}
Hence the partial rates are 
\begin{eqnarray}
& R_i=\frac{1}{d^{2^{k-i+1}}} \nonumber\\
& \sum_{l}{2^{k-i+1} \choose l} (d/2)^{l}(d/2)^{(2^{k-i+1}-l)} \nonumber\\
& \log({2^{k-i+1} \choose l} (d/2)^{l}(d/2)^{(2^{k-i+1}-l)}).
\label{MA:19}
\end{eqnarray}
As in the previous example the total rate of the protocol is given by Eq. \ref{MA:4} with the sum extended from $1$ to $k$ and  $R_{k+1}=0$.

Both protocols also work for states given by Eqs. \ref{MA:8} and \ref{MA:15} with a pure state $|\Phi_d^+\rangle\langle\Phi_d^+|_{AB}$ replaced by a mixed state $\sum_{k=0}^{d-1}q_kU^{k}_{A}|\Phi_d^+\rangle\langle\Phi_d^+|_{AB}U^{k\dagger}_{A}$, where $U^{k}_{A}=\sum_{l=0}^{d-1}\exp(\frac{i2\pi k l}{d})|l\rangle\langle l|_{A}$. The partial rates are given by Eqs. \ref{MA:14} and \ref{MA:19} with logarithms replaced by coherent information of the post-measurement state.

\section{Calculation of coherent information}

\subsection{Formulation of the problem}
\label{subsec:form}

We want to calculate coherent information of a state
\be
\rho_{lAB}^{(n)}=\frac{P_{lA}\otimes P_{lB}\rho_{AB}^{\ot n} P_{lA}\otimes P_{lB}}{
\text{Tr}(P_{lA}\otimes P_{lB}\rho_{AB}^{\ot n} P_{lA}\otimes P_{lB})}
\ee
Let us write $\rho^{\ot n}$ in the following form
\begin{eqnarray}
& \rho^{\otimes n}=x^{n}\rho_{AB}'^{\otimes n}+\nonumber \\
& x^{n-1}(1-x)[\rho_{AB}'^{\otimes(n-1)} |01\rangle\langle 01|_{AB}+\text{permutations}]+\nonumber \\
& +x^{n-2}(1-x)^{2}[\rho_{AB}'^{\otimes(n-2)}|01\rangle\langle 01|_{AB}^{\otimes 2}+\text{permutations}]\nonumber\\
& \dots+(1-x)^{n}|01\rangle\langle 01|_{AB}^{\otimes n},
\end{eqnarray}
where
\be
\rho'_{AB}=q|\Phi^{+}(\alpha)\rangle\langle\Phi^{+}(\alpha)|_{AB}+(1-q)|\Phi^{-}(\alpha)\rangle\langle\Phi^{-}(\alpha)|_{AB}. 
\ee
As noted before terms containing $|01\rangle\langle01|$ are not in the subspace on which Alice and Bob project the state. Hence, we have
\be
\rho_{lAB}^{(n)}=\frac{P_{lA}\otimes P_{lB}\rho_{AB}'^{\ot n} P_{lA}\otimes P_{lB}}{
\text{Tr}(P_{lA}\otimes P_{lB}\rho_{AB}'^{\ot n} P_{lA}\otimes P_{lB})}
\ee
Because the state of Bob's subsystem is an equal mixture of all standard basis vectors having $l$ $1$'s and $n-l$ $0$'s its entropy is equal to
\be
S(\rho^{n,l}_B)=\log{n \choose l}.
\ee
In order to calculate entropy of the whole system we note that it is equal to entropy of a simpler state (we denote it by $\rho_l^{(n)}$ without subscript $AB$)
\be
\rho_l^{(n)}=P_l\rho^{\ot n} P_l. 
\ee
where
\be
\rho=p|+\>\<+|+(1-p) \frac{I}{2}
\ee
with $|\pm\>=\frac{1}{\sqrt{2}}(|0\>\pm|1\>)$ and $p=2q-1$
Now our task is to find eigenvalues (together with their multiplicities) of the following
(subnormalized) state
\be
\rho_l^{(n)}=P_l\rho^{\ot n} P_l. 
\ee

\subsection{Statement of the main result}
\label{subsec:main}

Before we formulate the main result, we will prove the following lemma:
\begin{lemma}
\label{lem:1}
The matrix $\rho_l^{(n)}$ can be written as follows
\be
\rho_l^{(n)}=\frac{1}{2^n}\sum_k p^{2k} A_k^{(l)}
\label{eq:rho-A}
\ee
where $A_k^{(l)}$ is operator acting on the Hilbert space $\cH_l^{(n)}$ given by
\be
A_k^{(l)} = \sum_{x,y:d(x,y)=2k} |x\>\<y|
\ee
Here $|x\>,|y\>$ are vectors from $\cH_l^{(n)}$,
and $d(x,y)$ is Hamming distance between the binary sequences $x$ and $y$.
\end{lemma}

\begin{proof}
The initial state of $n$ particles is
\begin{eqnarray}
\left[\frac{p}{2}(|0\rangle+|1\rangle)(\langle0|+\langle1|)+\frac{1-p}{2}(|1\rangle\langle1|+|0\rangle\langle0|)\right]^{\otimes n}
\end{eqnarray}
and we project it on a subspace with definite number of $1's$. Hence both the state and the measurement operators are permutationally invariant. We can substitute
\be
\begin{split}
 P_{00}&=|0\rangle \langle 0|,\quad
 P_{01}=|0\rangle \langle 1| \\
 P_{10}&=|1\rangle \langle 0| ,\quad
 P_{11}=|1\rangle \langle 1|
\end{split}
\ee
and obtain the following expression which corresponds to the initial state
\begin{eqnarray}
& \left[\frac{p}{2}(P_{01}+P_{10})+\frac{1}{2}(P_{11}+P_{00})\right]^{\otimes n}=\nonumber\\
& =(\frac{1}{2})^{n}\sum_{k=0}^{n}p^{k}\hat{\sfS}[(P_{01}+P_{10})^{\otimes k}(P_{11}+P_{00})^{\otimes (n-k)}],\nonumber\\
\end{eqnarray}
where $\hat{\sfS}[...]$ denotes symmetrization. Here by symmetrization we mean the sum of all different permutations, e.g. $\hat{S}[a^{\otimes 2}\otimes b]=a\otimes a\otimes b+a\otimes b\otimes a+b\otimes a\otimes a$.
We are interested in a coefficient of $p^{k}$ which we can write as
\begin{eqnarray}
& \hat{\sfS}[(P_{01}+P_{10})^{\otimes k}(P_{01}+P_{10})^{\otimes (n-k)}]=\nonumber\\
&=\sum_{i=0}^{k}\sum_{j=0}^{n-k}\hat{\sfS}[P_{01}^{\otimes i}\otimes P_{10}^{\otimes (k-i)}\otimes P_{11}^{\otimes j}\otimes P_{00}^{\otimes (n-k-j)}].
\end{eqnarray}
The Hamming weight of this coefficient is determined by $P_{01}^{i}\otimes P_{10}^{k-i}$ and is equal to $k$. Because we project on subspace with $l$ $1$'s and $n-l$ $0$'s we have after the projection $i+j=l$ and $k-i+j=l$. Hence, all terms with odd $k$ vanish and moreover we can write $k=2i$.
\end{proof}

Here is the theorem which provides formula for eigenvalues of the
matrix $\rho_l^{(n)}$.
\begin{theorem}
\label{thm:waw}
The eigenvalues of $\rho_l^{(n)}$ are given by
\be
\lambda^{n}_l(j)=\frac{1}{2^n} \sum_{k=0}^lp^{2k} \alpha_k(j),
\ee
where $j=0,\ldots,\min\{l,n-l\}$, and $\alpha_k(j)$ are
eigenvalues of operators $A_k^{(l)}$.
Two alternative forms of those eigenvalues are   given by theorems
\ref{thm:eig-A-Young} and  \ref{thm:eig-A-MM} below.
The multiplicities of the eigenvalues are given by
\be
f_j=\binom{n}{j}\frac{n-2j+1}{n-j+1}.
\ee
\end{theorem}

Here are our alternative formulas for eigenvalues of operators $A_k^{(l)}$:
\begin{theorem}
\label{thm:eig-A-Young}
The eigenvalues of operator $A_k^{(l)}$ have following form
\begin{equation}
\label{eq:eig-A-Young}
\alpha_k(j)=\sum_{r=min\{k,l-j\}}^{r=max\{0,k-j\}}(-1)^{k-r}\binom{l-j+k-r}{ k}\binom{n-l-k+r}{r}\binom{j}{k-r}.
\end{equation}
\end{theorem}
We  shall prove this form by use of Young diagrams.

Using so-called algebraic association schemes, we obtain another expression for eigenvalues:

\begin{theorem}
\label{thm:eig-A-MM}
The eigenvalues of operator $A_{k}^{(l)}$ have the following form%
\[
\alpha _{k}(j)=\sum_{r=0}^{k}(-1)^{k-r}
{l-r \choose k-r}{l-j \choose r} {n-l-j+r \choose r} \equiv E_{k}(j)
\]%
where
\[
E_{k}(u)=(-1)^{k}{l\choose k} \quad _{3}F_{2}\left(
\begin{array}{ccc}
-k, & -l+u, & n-l-u+1 \\
-l, & 1, &
\end{array}%
;1\right)
\]%
is the dual Hahn polynomial and is the hypergeometric function.
\end{theorem}

At the end of this section we prove explicit formula for spectral radius of matrices $\rho_l^{(n)}$. We show also that maximal eigenavalue of $\rho_l^{(n)}$ is always is smaller than $1$.
\begin{lemma}
\label{lem:specRadius}
 The spectral radius $\lambda _{0}$ of the matrix $\rho_{l}^{(n)}$ is the
following
\be
\lambda _{0}=\frac{1}{2^{n}}\sum_{k=0}^{l}\mathcal{P}^{\;2k}
{n-l \choose k}{l\choose k},\quad \lambda _{0}<1,
\ee
\ \ and it is an eigenvalue of $\rho_{l}^{(n)}$ with the algebraic multiplicity $1$.
\end{lemma}

\begin{proof}
 For a given basis vector $e_{i}$ any other basis vector $e_{j}$ is at a
Hamming ditance $2k$ for some $k=0,1,...,l.$ From the lemma~\ref{lem:dist} (see Appendix) it follows
that there are ${n-l\choose k}{l\choose k}$  of them and this number does not depend on a 
given basis vector $e_{i}$.

 It is easy to see  that the vector $(1,1,....,1)\in
\mathbb{C}
^{\dim \rho_{l}^{(n)}}$ is an eigenvector of $\rho_{l}^{(n)}$ with eigenvalue $%
\lambda _{0}$ which is simply the sum of all elements in each row of $%
\rho_{l}^{(n)}$, i.e.%
\be
\lambda _{0}=\frac{1}{2^{n}}\sum_{k=0}^{l}\mathcal{P}^{\;2k}
{n-l \choose k}{l\choose k}
\leq \frac{1}{2^{n}}\sum_{k=0}^{l}
{n-l\choose k}{l\choose k}
=\frac{1}{2^{n}}
{n \choose l} <1
\ee%
where we have used that $\mathcal{P}=2p-1\leq 1$ if $p\in \lbrack 0,1].$

The fact that $\lambda _{0}$ is a spectral radius of algebraic multiplicity $%
1$ follows from the basic theorem on stochastic matrices~\cite{Kostrikin}.
\end{proof}

\subsection{Mathematical introduction}
\label{subsec:math_intr}
\subsubsection{Schur-Weyl decomposition and Young diagrams}
Now we shall use a couple of facts about the following unitary representation of permutation
group $S_n$ on $(\C^d)^{\ot n}$: for given permutation $\pi$
a unitary $V_\pi$ is given by
\be
V_\pi |i_1\>\ot \ldots \ot |i_n\> = |i_{\pi(1)}\>\ldots |i_{\pi(n)}\>.
\label{eq:swaps}
\ee
Here $|i_1\> \ldots |i_n\>$ is standard basis in $(\C^d)^{\ot n}$,
where $i_j=1,\ldots,d$.
The notation is mostly taken from \cite{Audenaert2006-notes}.
The space $(\C^d)^{\ot n}$ can be decomposed into irreducible representations of $S_n$
\be
(\C^d)^{\ot n}=\oplus_\lambda\hcal^U_\lambda \ot {\hcal}^S_\lambda
\label{eq:Schur-Weyl}
\ee
where $\lambda$ labels inequivalent irreps of $S_n$, and $\hcal^U_\lambda$ is
multiplicity space (the label $U$ comes form the fact, that it is at the same time
irrep of unitary group $U(d)$).  It is called Schur-Weyl decomposition.
The labels $\lambda$ are {\it partitions}
of the set $\{1,\ldots, n\}$. Partition is a sequence $\lambda=(\lambda_1,\ldots,\lambda_s)$
of nonnegative integers satisfying
\ben
&&\lambda_1\geq \lambda_2 \geq \ldots\geq \lambda_s \nonumber\\
&& \sum_{i=1}^s \lambda_i= n
\een
where $s\in\{1,\ldots, n\}$. The direct sum \eqref{eq:Schur-Weyl} runs over
all partitions $\lambda$ with $s\leq d$. The partitions
can be represented by means of diagrams, and are then called Young diagrams.
Here are few examples with  corresponding partitions $\lambda$.

\[
\begin{split}
&\yng(2,2)\qquad \qquad \yng(3,2,1)\qquad \qquad \yng(4)\qquad \qquad \yng(1,1,1)\\
\lambda&=(2,2),\qquad \lambda=(3,2,1),\qquad \  \  \lambda=(4),\qquad \lambda=(1,1,1)\\
\end{split}
\]

In our case $d=2$, hence $\lambda$ runs over binary partitions or, equivalently,
over Young diagrams with two rows.
Hence the partitions are of the form $(n-j,j)$ and they can be labeled by $j$, i.e., the length of the second row (note that $j\leq n/2$). Given Young diagram, one defines  {\it standard Young tableax} (SYT)
as a diagram filled with numbers $k\in\{1,\ldots, n\}$ in such a way that in each row,
the numbers strictly increase from left to right, and in each column
they strictly increase form top to bottom. The number of SYT's for a fixed diagram $\lambda$,
which we denote by $f_\lambda$ is equal to the dimension of the irrep labeled by $\lambda$.
In the case of binary partitions we have
\be
f_j:=f_{(n-j,j)} ={n \choose j}\frac{n-2j+1}{n-j+1}
\ee
With a given SYT $a$, one associates a so called Young symmetrizers $P^{\lambda,a}$, and
which are constructed from operators $\sfA_k$ and $\sfS_k$,
which are proportional to projectors onto completely antisymmetric and symmetric subspaces of
$(\C^d)^{\ot k}$, $1\leq k\leq n$
\be
\sfS_k=\sum_{\pi\in S_n} V_\pi,\quad  \sfA_k=\sum_{\pi\in S_n} (-1)^{{\rm sgn}(\pi)} V_\pi
\label{eq:SA}
\ee
Now for a fixed row of SYT, we consider operator $\sfS_k$ which acts on
the systems labeled by the numbers from the row. We extend it to the full system,
by multiplying with identities on other systems. Similarly with every column,
we associate operator $\sfA_k$. Now the Young symmetrizer is a product of three factors:
normalization constant  $\frac{f_\lambda}{n!}$, the product of $\sfA_k$'s over
all columns, and the product of $\sfS_k$'s over all rows:
\be
P^{\lambda,a} = \frac{f_\lambda}{n!} \mathop{\Pi}_{k\in Col(\lambda,a)} \sfA_k
\mathop{\Pi}_{k\in Row(\lambda,a)} \sfS_k
\label{eq:Ysym}
\ee
The symmetrizers are projectors, i.e. they satisfy $P^2=P$, but they are
usually not orthogonal projectors, i.e. they fail to satisfy $P^\dagger=P$.

Finally, we need to know, how the Young symmetrizers are related to the Schur-Weyl decomposition.
Namely, they are of the following form:
\be
P^{\lambda,a}= I^U_\lambda \ot |u\>\<v|.
\label{eq:Y-decomp}
\ee
where $|u\>,|v\>\in \hcal^S_\lambda$ and $I^U_\lambda$
is identity operator acting on the space $\hcal^U_\lambda$.

\subsubsection{Algebraic association schemes}
Here we recall some resutls from
theory of the algebraic association schemes~\cite{Bannai}.

\bigskip

\begin{definition}[B-I]
\label{def:BI1}
Let $X$ be a set of cardinality $n$ and let $R_{i}$, $i=0,1,...,d$ be
subsets of $X\times X$ With property that

(i) $R_{0}=\{(x,x),\quad x\in X\}$.

(ii) $X\times X$ $=\cup _{i=0}^{d}R_{i},\quad R_{i}\cap R_{j}=\varnothing$ if $i\neq j$.

(iii) $R_{i}^{t}=R_{i^{\prime }}$ for some $i^{\prime }\in \{0,1,...,d\}$
where $R_{i}^{t}=\{(x,y)\quad |\quad (y,x)\in R_{i}\}$.

(iv) For $i,j,k\in \{0,1,...,d\}$, the number of $z\in X$ such that $%
(x,z)\in R_{i}$ and $(z,y)\in R_{j}$ is constant whenever $(x,y)\in R_{k}.$%
This constant number is denoted $p_{ij}^{k}.$

(v) $p_{ij}^{k}=p_{ji}^{k}\quad \forall i,j,k\in \{0,1,...,d\}.$

Such a configuration $\Xi =(X,\{R_{i}\}_{i=0}^{d})$ is a Commutative
Association Scheme $(CAS)$ of class $d$. The non-negative integers $%
p_{ji}^{k}$ are called the intersection numbers. A $CAS$ with the additional property

(vi) $R_{i}^{t}=R_{i}$

is called a symmetric $CAS$.
\end{definition}

For any commutative association scheme one can define\newline

\begin{definition}[B-I]
\label{def:BI2}
The $k$'th adjacency matrix $A_{k}$ $\quad k\in \{0,1,...,d\}$ \ of $CAS\ $\ $%
\Xi =(X,\{R_{i}\}_{i=0}^{d})$ is a matrix of degree $|X|=n$ whose rows and
columns are indexed by the elements $X$ and whose entries are
\begin{equation}
(A_{k})_{(x,y)}=%
\begin{array}{c}
1\quad if\quad (x,y)\in R_{k} \\
0\quad if\quad (x,y)\notin R_{k}%
\end{array}%
.
\end{equation}
So $i$'th adjacency matrix $A_{k}$ is a $0,1$ matrix.
\end{definition}

It is easy to show that the defining conditions (i),...,(v) for $CAS$ are
equivalent to the following conditions (i'),...,(v') for the adjacency
matrices $A_{i}$ $\quad i\in \{0,1,...,d\}$

\begin{proposition}[B-I]
\label{prop:BI1}
$(X,\{R_{i}\}_{i=0}^{d})$The matrices $A_{i}$ $\quad i\in \{0,1,...,d\} $
are adjacency matrices for $CAS$ \ $\Xi =(X,\{R_{i}\}_{i=0}^{d})$ iff

(i') $A_{0}=\mathbf{1},$ the identity matrix.

(ii') $\sum_{k=0}^{d}.A_{k}=J,$ where $J$ is the matrix whose entries are
all $1.$

(iii') $A_{k}^{t}=A_{k^{\prime }}$ for some $k^{\prime }\in \{0,1,...,d\}.$

(iv') $A_{i}A_{j}=\sum_{k=0}^{d}p_{ij}^{k}A_{k}$ $\ \forall i,j,k\in
\{0,1,...,d\}.$

(v') $p_{ij}^{k}=p_{ji}^{k}\quad \forall i,j,k\in \{0,1,...,d\}$ \ $%
\Leftrightarrow $ \ $A_{i}A_{j}=$ $A_{j}A_{i}\quad \forall i,j\in
\{0,1,...,d\}.$

And for a symmetric $CAS$ we have

(vi') $A_{k}^{t}=A_{k}$ \ $\forall k\in \{0,1,...,d\}.$
\end{proposition}

\bigskip

\begin{theorem}
\label{thm:CCAS}
Suppose that $(X,\{R_{i}^{X}\}_{i=0}^{d})$ is a CAS and $Y$ is a set such
that there is a bijection $\varphi :X\rightarrow Y$. Then a pair $%
(Y,\{R_{i}^{Y}\}_{i=0}^{d})$ where
\be
R_{i}^{Y}\in Y\times Y,\qquad R_{i}^{Y}=\{(y,y^{\prime })\quad |\quad
(\varphi ^{-1}(y),\varphi ^{-1}(y^{\prime }))\in R_{i}^{X}\}\equiv \Phi
(R_{i}^{X})\}
\ee
is a CAS and its adjacency matrices are equal to adjacency matrices of the CAS $%
(X,\{R_{i}\}_{i=0}^{d}).$
\end{theorem}

\begin{proof}
A pair $(Y,\{R_{i}^{Y}\}_{i=0}^{d})$ is a CAS because the set $Y$ and the
family of sets $\{R_{i}^{Y}\}_{i=0}^{d}$ are bijective images of $X$ and $%
\{R_{i}^{X}\}_{i=0}^{d}$ respectively. Let us prove that  the adjacency matrices
in these CAS are equal. We denote by $\{A_{i}^{X}\}_{i=0}^{d}$ ($%
\{A_{i}^{Y}\}_{i=0}^{d}$) the adjacency matrices of the CAS $%
(X,\{R_{i}^{X}\}_{i=0}^{d})$ ($(Y,\{R_{i}^{Y}\}_{i=0}^{d})$) respectively.
Then we have%
\be
(A_{i}^{Y})_{(yy^{\prime })}=%
\begin{array}{c}
1\quad if\quad (y,y^{\prime })\in R_{i}^{Y} \\
0\quad if\quad (y,y^{\prime })\notin R_{i}^{Y}%
\end{array}%
\Leftrightarrow
\begin{array}{c}
1\quad if\quad (\varphi ^{-1}(y),\varphi ^{-1}(y^{\prime }))\in R_{i}^{X} \\
0\quad if\quad (\varphi ^{-1}(y),\varphi ^{-1}(y^{\prime }))\notin R_{i}^{X}%
\end{array}%
=
\ee
\[
=%
\begin{array}{c}
1\quad if\quad (x,x^{\prime })\in R_{i}^{X} \\
0\quad if\quad (x,x^{\prime })\notin R_{i}^{X}%
\end{array}%
=(A_{i}^{X})_{(x,x^{\prime })},
\]
where $y=\varphi (x)$ and $y^{\prime }=\varphi (x^{\prime }),$ i.e. we have
\be
(A_{i}^{X})_{(x,x^{\prime })}=(A_{i}^{Y})_{(\varphi (x),\varphi (x^{\prime
}))}.
\ee
\end{proof}

\bigskip

The most important, for our paper, example of $CAS$ is the following

\bigskip

\begin{proposition}[B-I]
\label{prop:BI2}
Let $V$ be a set of cardinality $n$ and let $l$ be a non-negative integer
such that $l\leq \frac{n}{2}.$ Let $X^{J}$ be a set of $l$-element subsets
of $V$, so that $|X^{J}|={n \choose k}$.  Define%
\be
R_{k}^{J}=\{(x,y)\quad |\quad x,y\in X^{J},\quad |x\cap y|=l-k\}.\quad
k=0,1,..,l.
\ee
Then the pair $\Xi ^{J}=(X^{J},\{R_{i}^{J}\}_{i=0}^{l})$ is a symmetric $CAS$
of class $l$ called Johnson scheme. The corresponding adjacency matrices $%
A_{k}^{J}$ \ $\forall k\in \{0,1,...,l\}$ have the following eigenvalues%
\be
\alpha _{k}(j)=\sum_{r=0}^{k}(-1)^{k-r}
{l-r \choose
k-r}
{l-j \choose k-r}
{n-l-j+r\choose r}
\equiv E_{k}(j)
\ee%
where $j=0,1,...,l$ and it labels the common eigenspaces of $A_{k}^{J}$ $($
all $A_{k}^{J}$ \ $\forall k\in \{0,1,...,l\}$ commute$)$ and where
\be
E_{k}(u)=(-1)^{k}
{l \choose k}
\quad _{3}F_{2}\left(
\begin{array}{ccc}
-k, & -l+u, & n-l-u+1 \\
-l, & 1, &
\end{array}%
;1\right)
\ee
 is the dual Hahn polynomial and $F$ is the hypergeometric function.

\end{proposition}

\begin{remark}
If we describe the set $V$ as $V=\{1,2,...,n\}$, then any $l$-element subset
of $V$, i.e. the element of $X^{J},$ may be denoted in a natural way by $%
x\{i_{1},i_{2},...,i_{l}\}\equiv x\{i\}$ where $i_{1},i_{2},...,i_{l}$
denotes the elements of $V=\{1,2,...,n\}$ which are contained in the subset $%
x\{i_{1},i_{2},...,i_{l}\}\in X^{J}$\bigskip .
\end{remark}

\subsection{Proofs of main results}
\label{subsec:proof}
\subsubsection{Some facts about space $\cH_l^{(n)}$}

We have the following lemma:

\begin{lemma}
\bei
\item[(i)] The Hamming distance $d(x,y)$ between
two bit-strings $x$ and $y$ is even and satisfies $d(x,y)\leq \min\{l,n-l\}$.
\label{it:even}
\item[(ii)]
For any two pairs of vectors
$(x,y)$ and $(x',y')$ such that $d(x,y)=d(x',y')$  there exists permutation $\sigma$ such that
$(\sigma(x),\sigma(y))=(x',y')$.
\label{it:trans}
\item[(iii)] The operators $A_k$ mutually commute.\\
\label{it:comm}
\item[(iv)] Any operator acting on $\cH_l^{(n)}$ which is invariant under permutations of
qubits is a linear combination of those operators.
\label{it:inv}
\eei
\label{lem:A-prop}
\end{lemma}

{\slshape Proof of (i)}\\
For a given pair $(x,y)$,
let us divide $x$ into two parts: the first one consists of positions, where $x$ and $y$
agree, and the second  one consists of positions, where they disagree.
Then $d(x,y)$ is the length of the latter part.
Since the number of $1$'s in $x$ and $y$ is equal,
and in the first part, by definition, it is also equal, then also
in the second part the number of $1$'s (and therefore also $0$'s)  is equal.
It follows that the length of the second part is even, and also it cannot be greater
than the total number $l$ of $1$'s in $x$  and than the total number of $0$'s in $x$(which is equal to $n-l$)
Thus, in particular, $2k=d(x,y)$ is an even number.
$\qed$

{\slshape Proof of (ii)}\\
Let us consider the partition  of $x$ into two parts as in the proof of item (i).
Let us further apply to $x$ and $y$ permutation, which moves all bits of the second part to the right, and then in each part of $x$ 
moves $1$'s to the right. Here is an example:
\be
\bea{l}
x=(010011011)\\
y=(001010111)
\eea \to
\bea{l}
(00111|1010)\\
(00111|0101)
\eea \to
\bea{l}
(00111|0011)\\
(00111|1100)
\eea
\ee
In this way we have transformed $(x,y)$ into $(x^0,y^0)$
where $x^0$ and $y^0$ are a kind of canonical vectors, which depend only
on $n,l$ and $k=d(x,y)/2$:
\ben
&&x^0=\underbrace{0\ldots0}_{n-l-k} \underbrace{1\ldots1}_{l-k}
\underbrace{0\ldots0}_{k} \underbrace{1\ldots1}_{k}\nonumber\\
&&y^0=\underbrace{0\ldots0}_{n-l-k} \underbrace{1\ldots1}_{l-k}
\underbrace{1\ldots1}_{k} \underbrace{0\ldots0}_{k}
\een
Let us call the permutation $\sigma_{xy}$.
It follows that if $d(x,y)=d(x',y')$, then $(\sigma(x),\sigma(y))=(x',y')$
with $\sigma=\sigma_{xy} \sigma_{x'y'}^{-1}$.
$\qed$

{\slshape Proof of (iii)}\\
Let us prove that $A_k$ commute. We directly check that
\be
A_k A_{k'} =\sum_{x,z} f_{kk'}(x,z) |x\>\<z|, A_{k'} A_k=\sum_{x,z} f_{kk'}(x,z) |x\>\<z|
\ee
where $f_{kl}(x,z)=|y: d(x,y)=2k, d(y,z)=2l|$.
It is enough to show that for any pair $(x,z)$ (recall that $wt(x)=wt(z)$
where $wt$ stands for weight, i.e. the number of $1's$)
we have
\be
f_{kk'}(x,z)=f_{k'k}(x,z).
\label{eq:ff}
\ee
We shall now establish a reversible mapping, which for fixed $(x,z)$
will map any $y$ satisfying $d(x,y)=2k,d(y,z)=2k'$ into  $y'$
satisfying $d(x,y')=2k',d(y',z)=2k$.
This would prove, that the number of $y$'s is the same as
the number of $y'$'s, hence \eqref{eq:ff} holds.
Let us now describe the mapping - call it g. Its action is to
flip all bits, where $x$ and $z$ differ.  Note that $g$ is its own inverse.
We now notice that $g(x)=z$ and $g(z)=x$.
We set $y'=g(y)$. Since $d(a,b)=d(g(a),g(b))$,
we have $d(x,y')=d(g(x),y)= d(z,y)=d(y,z)$ and similarly $d(y',z)=d(x,y)$.
$\qed$

{\slshape Proof of (iv)}\\
Consider arbitrary operator
\be
C=\sum_{x,y} c_{xy} |x\>\<y|
\ee
which is invariant under permutation of qubits, i.e. for any permutation $\pi$
we have
\be
V_\pi C V_\pi^\dagger = C
\label{eq:c-inv}
\ee
Let us first argue, that if two pairs $(x,y)$ and $(x',y')$
can be joined by some permutation $\sigma$ (i.e. $\sigma(x)=x'$ and $\sigma(y)=y'$,
then $c_{xy}=c_{x'y'}$. Let us note that
\be
V_{\sigma^{-1}} C V_{\sigma^{-1}}^\dagger =\sum_{x,y} c_{x,y}|\sigma^{-1}(x)\>\<\sigma^{-1}(y)|= \sum_{x,y} c_{\sigma(x),\sigma(y)}
|x\>\<y|
\ee
Thus from \eqref{eq:c-inv} we get
\be
c_{x,y}=c_{\sigma(x),\sigma(y)}
\ee
Now, from (ii) we know, that  if two pairs have the same Hamming
distance, they can be joined by a permutation in the above sense.
Thus $c_{xy}$ is constant on pairs that have a fixed Hamming distance,
which proves that $C$ is a linear combination of operators $A_k$. $\qed$

Since our subspace  $\cH_l^{(n)}$ is invariant under permutations of systems,
it is a subrepresentation of the representation \eqref{eq:swaps} of $S_n$ in $(\C^d)^{\ot n}$.
It turns out, that if we decompose $\cH_l^{(n)}$ into irreps of $S_n$, each irrep will appear
at most once. Namely we have
\begin{lemma}
The space $\cH_l^{(n)}$ has the following decomposition into irreps of
$S_n$:
\be
\cH_l^{(n)}=\oplus_{j=0}^{\min\{l,n-l\}} B_j
\ee
where $B_j$ is irrep labeled by partitions $\lambda=(n-j,j)$.
Moreover the operators $A_k^{(l)}$ have the following form
\be
A_k^{(l)}=\sum_{j=0}^{\min\{l,n-l\}} \alpha_k(j) P_j^{(l)}
\label{eq:A-spectral}
\ee
where $P_j^{(l)}$ are projectors which project onto irreps $B_j$.
\label{lem:decomp}
\end{lemma}
\begin{remark}
 This means, in particular, that
\be
\hcal^U_\lambda \ot \hcal^S_\lambda \cap \hcal_l^{(n)} = \C \ot \hcal^S_\lambda
\ee
for $\lambda$ corresponding to irreps which appear in the decomposition of $\hcal_l^{(n)}$.
\end{remark}

{\bf Proof.}
The items (iii) and (iv) of Lemma \ref{lem:A-prop} imply, that the set of all the operators
acting on $\hcal_l^{n}$ which are invariant under permutations is a commuting set. Therefore in the decomposition $\hcal_l^{n}$ into irreps, the multiplicity spaces have to be trivial. That $j$ must be no greater than $l$ and than $n-l$
follows from the fact  that Young symmetrizes corresponding to irrep $(n-j,j)$ kill vectors with less  number of $1$'s than $j$ and also those with less $0$'s than $j$. This fact is easy to directly verify
basing on properties of operators \eqref{eq:SA}.  (It is actually related to construction of the basis in multiplicity space of irreps of $S_n$ via so-called semi-standard Young tableaux.)

In Appendix we present and alternative proof which follows from explicit formula for characters  of irreps of $S_n$ labeled by two-row Young diagram.

\subsubsection{Proof of Theorem \ref{thm:eig-A-Young}}
\label{subsec:eig-A-Young}
Due to \eqref{eq:A-spectral}, to compute eigenvalues, it is enough to take arbitrary (not necessarily
normalized) vector $|\psi\>$ from $B_j$.
Then we have
\be
\alpha_k(j) = \frac{\<\psi|A_k^{(l)}|\psi\>}{\<\psi|\psi\>}
\label{eq:alpha_psi}
\ee

Consider now the following tableau
\be
((1,\ldots, n-j),(n-j+1,\ldots,n)),
\label{eq:tab}
\ee
and
denote the related symmetrizer by $\tilde P_j$ (explicitly $\tilde P_j$
is equal to $P^{\lambda,a}$ of \eqref{eq:Ysym} with $\lambda=(j,n-j)$ and $a$ being the
above tableau).

If we now take a vector
\be
|x_0\> = |\underbrace{0\ldots 0}_{n-l}\underbrace{1\ldots 1}_{l}\>
\label{eq:x_zero}
\ee
then
\bei
\item[(i)] The vector $\tilde P_j |x_0\>$ is nonzero (provided $j\leq \min\{l,n-l\})$
\item[(ii)] The vector $\tilde P_j |x_0\>$ belongs to $B_j$.
\eei
The first fact is to easy verify directly by using explicit form for Young symmetrizers \eqref{eq:Ysym}. (It is related to the fact, that the number of semi-standard
Young tableaux is equal to dimension of multiplicity space). To prove the second item,
note that the projection onto $B_j$ which we shall denote by $P_j$ is of the following form
with respect to the decomposition \eqref{eq:Schur-Weyl}
\be
P_j=|w\>\<w|\ot P_j^S\, ,\quad |w\>\in\hcal_\lambda^U,
\label{eq:proj-lambda}
\ee
where $P_j^S$ projects onto the space $\hcal^S_\lambda$ of the decomposition,
with $\lambda=(j,n-j)$.
Thus, the projector $P_j$ projects onto a subspace being intersection of $\hcal^U_\lambda \ot\hcal^S_\lambda$ and $\hcal_l^{(n)}$, with this $\lambda$.
Let us denote $|\phi\>=\tilde P_j |x_0\>$. Let us aruge, that $\phi$ indeed belongs to 
those two spaces. 
On one hand, $\phi$ belongs to $\hcal_l^{(n)}$, because $x_0$ does, 
and the symetrizers are build out of permutation operators \eqref{eq:swaps}, which leave
the subspace $\hcal_l^{(n)}$ invariant. On the other hand, it belongs to 
$\hcal^U_\lambda \ot\hcal^S_\lambda$ as follows from the form of Young symmetrizers 
given by \eqref{eq:Y-decomp}.


Thus $\tilde P_j x_0$ can be taken as a vector  to be inserted  into \eqref{eq:alpha_psi}.
Doing this, and using the fact that operators $A_k$ commute with permutations
we obtain the following lemma:
\begin{proposition}
For fixed $k$ we have
\begin{equation}
\label{lambda1}
\alpha_k(j)=\frac{\langle x_0|\operatorname{\sfS}\operatorname{\sfA}A_k\operatorname{\sfS}|x_0\rangle }{\langle x_0|\operatorname{\sfS}\operatorname{\sfA}\operatorname{\sfS}|x_0\rangle},
\end{equation}
where $j\in\{0,1,...\min\{l,n-l\}\}$ labels all allowed partitions, $\sfS$ and $\sfA$ are symmetrizers for  partition labeled by $j$, and $|x_0\>=|\underbrace{0\ldots 0}_{n-l}\underbrace{1\ldots 1}_{l}\>$.
\end{proposition}

\begin{proof}[of theorem \ref{thm:eig-A-Young}]
Now we want to calculate explicit combinatorial formula for $\alpha_k(j)$ which depends only on given partition $j$, number of zeros $n-l$ and number of ones $l$. Before we do it, notice  that all operators $A_k$ commute with all operators $\operatorname{\sfA}$ and $\operatorname{\sfS}$, then
\begin{equation}
\langle x_0| \operatorname{\sfS\sfA}A_{k}\operatorname{\sfS}|x_0\rangle=\langle x_0|\operatorname{\sfS\sfA\sfS}A_k|x_0\rangle=\langle x_0|\operatorname{\sfS\sfA\sfS}|A_kx_0\rangle=\sum_{y\in\mathcal{Y}_{|x_0\rangle}}\langle x_0|\operatorname{\sfS\sfA\sfS}|y\rangle,
\end{equation}
where $\mathcal{Y}_{|x_0\rangle}:=\{|y\rangle \in H_l^{(n)} \ | \ d(|x_0\rangle,|y\rangle)=2k\}$. So we can rewrite equation~\eqref{lambda1} in a form
\begin{equation}
\alpha_k(j)=\sum_{y\in\mathcal{Y}_{|x_0\rangle}}\frac{\langle x_0|\operatorname{\sfS}\operatorname{\sfA}\operatorname{\sfS}|y\rangle }{\langle x_0|\operatorname{\sfS}\operatorname{\sfA}\operatorname{\sfS}|x_0\rangle}.
\label{eq:alfa1}
\end{equation}
Notice that operator $\operatorname{A}_k$ acts in fixed subspace $H_l^n$, so the number of $1$'s in vectors $|x_0\>$ and $|y\>=\operatorname{A}_k|x_0\>$ is the same.

Since the operators $\sfS$ and $\sfA$ are constructed with respect to
our chosen tableau \eqref{eq:tab}, it is convenient to decompose any given vector $|y\>$
into three parts related to the tableaux.

\begin{equation}
\begin{tabular}{l|c|c|}\cline{2-3}
{\small the first row}$ \  \  \quad \rightarrow$&$ \qquad y_1 \qquad $ & $ \qquad y_3 \qquad $ \\
\cline{2-3}
{\small the second row}$ \ \rightarrow$&$ \qquad y_2 \qquad $\\
\cline{2-2}
\end{tabular}
\label{eq:3-parts}
\end{equation}
i.e. $|y\>=|y_1\>_1|y_2\>_2|y_3\>_3$.

As we will prove in Lemma \ref{lem:ASA} the elements of sum \eqref{eq:alfa1}
depend on $y$ only through the number of $1$'s in the first (or, equivalently, the second) row.
Thus it is convenient to  partition the set $\mathcal{Y}_{|x_0\rangle}$ into smaller sets defined
as
\be
\mathcal{Y}_{|x_0\rangle}^m=
\{|y\rangle \in H_l^{(n)}: d(x_0,y)=2k, wt(y_1y_3)=m\}
\ee
where $wt(x)$ denotes the number of $1$'s in $x$.
We can then rewrite \eqref{eq:alfa1} as follows
\be
\alpha_k(j)=\sum_{m=\operatorname{max}\{k,l-j\}}^{\operatorname{min}\{l,l-j\}}\sum_{y\in\mathcal{Y}_{|x_0\rangle}^m}\frac{\langle x_0|\operatorname{\sfS}\operatorname{\sfA}\operatorname{\sfS}|y\rangle }{\langle x_0|\operatorname{\sfS}\operatorname{\sfA}\operatorname{\sfS}|x_0\rangle}
=\sum_{m=\operatorname{max}\{k,l-j\}}^{\operatorname{min}\{l,l-j\}} |\mathcal{Y}_{|x_0\rangle}^m|
\frac{\langle x_0|\operatorname{\sfS}\operatorname{\sfA}\operatorname{\sfS}|y_0\rangle }{\langle x_0|\operatorname{\sfS}\operatorname{\sfA}\operatorname{\sfS}|x_0\rangle}.
\label{eq:alfa2}
\ee
where $y_0$ is arbitrary vector with $m$ $1$'s in upper row (see \eqref{eq:3-parts}).

Inserting the values of $|\mathcal{Y}_{|x_0\rangle}^m|$ and
$\langle x_0|\operatorname{\sfS}\operatorname{\sfA}\operatorname{\sfS}|y_0\rangle /\langle x_0|\operatorname{\sfS}\operatorname{\sfA}\operatorname{\sfS}|x_0\rangle$
obtained in L emmas \ref{l1} and \ref{lem:ASA}, respectively, we obtain
the formula
\be
\label{eq:for1}
\alpha_k(j)=\sum_{m=\operatorname{max}\{k,l-j\}}^{\operatorname{min}\{l,l-j\}}(-1)^{j-l+m}\frac{{n-j \choose n-l}}{{n-j \choose m}}{n-l \choose m-k}{l-j \choose m-k}{j \choose l-m}.
\ee
Now using transformation $m=l-j+k-r$ we obtain equation~\eqref{eq:eig-A-Young}. This ends the proof of Theorem \ref{thm:eig-A-Young}.
\end{proof}

Here we present the two lemmas used in the above proof.

\begin{lemma}
\label{l1}
The set  $\mathcal{Y}_{|x_0\rangle}^m$ of all vectors $y$ which have $m$ $1$'s
in the first row, and which satisfy $d(x_0,y)=2k$ has the number of elements given by
the following formula:
\begin{equation}
\label{card}
|\mathcal{Y}_{|x_0\rangle}^m|={ n-l \choose k } {l-j \choose m-k } { j \choose
l-m }.
\end{equation}
\end{lemma}

\begin{proof}

The vector $|x_0\>$ if inscribed into our Young diagram looks as follows:
\begin{equation}
\label{p1}
\begin{split}
\text{first row}\rightarrow&\text{\phantom{.}}0\cdots 0| \cdots 0 \ 1\cdots 1 \\
\text{second row}\,\rightarrow&\underbrace{1\cdots 1|}_j\\
&\qquad \quad \downarrow \ \operatorname{A}_k \\
&\text{\phantom{.}}|0\cdots  \cdots \cdots\ \cdot \overbrace{1 \ 1}^{l_1}|\cdots  \ 0\cdots 0 \ \overbrace{1\cdots 1|}^{l_3}\\
&\underbrace{|1\cdots 0\cdots 0 \ 1\cdots \ 1|}_j
\end{split}
\end{equation}

Let us denote by $m=l_1+l_3$  the number of ones in the first row of vector $|x_0\rangle$ and by $l_2$ the number of ones in the second row also in $|x_0\>$.  Hamming distance $2k$ is composed by two "subdistances" $k_1$ and $k_2$, so $2k=k_1+k_2$. Number $k_1$ corresponding to Hamming distance between the first row of vector $|x_0\rangle$ and first row of vector $|y\rangle$, $k_2$ corresponding to Hamming distance between first rows of our vectors. We also denote by $l_3$ the number of ones in the first row of vector $|y\>$  which overlap with ones in the first row of vector $|x_0\>$ and by $l_1$ the number of ones in the first row of vector $|y\>$ which overlap with zeros in the first row of vector $|x_0\>$. Thanks to this we can find
\begin{equation}
k_1=l-j+l_1-l_3,\qquad k_2=j-l_2,
\end{equation}
so Hamming distance $k$ is equal to
\begin{equation}
\label{h1}
2k=k_1+k_2=l+l_1-l_3-l_2=l+2l_1-\underbrace{(l_1+l_2+l_3)}_{l}=2l_1.
\end{equation}
Number of permutations in the second row preserving distance $k_2$ is equal to ${j\choose l_2}$. Number of permutations in a first row preserving distance $k_1$ is equal to ${n-l\choose l_1}{l-j\choose l_3}$. Finally using equations \eqref{h1}, $l=m+l_2$ and $l_3=m-k$ we obtain cardinality of the set $\mathcal{Y}_{|x_0\rangle}$:
\begin{equation}
|\mathcal{Y}_{|x_0\rangle}|={n-l\choose k}{l-j\choose m-k}{j\choose l-m}.
\end{equation}
\end{proof}

\begin{lemma}
\label{lem:ASA}
Let $y$ have $m$ $1$'s in upper row. Then
\be
\label{eq:asa}
\frac{\langle x_0|\operatorname{\sfS}\operatorname{\sfA}\operatorname{\sfS}|y\rangle }{\langle x_0|\operatorname{\sfS}\operatorname{\sfA}\operatorname{\sfS}|x_0\rangle}
=(-1)^{j-l+m}\frac{{n-j \choose l-j}}{{n-j \choose m}}
\ee
\end{lemma}

\begin{proof}
Consider fixed vector $|x_0\rangle$ and an arbitrary vector $|y\rangle$ with $l$ ones. Above-mentioned vectors can be decomposed into Young diagrams of shape $(n-j,j)$ like in lemma~\ref{l1}.  In every such diagram we isolated three parts like on  picture below:

\begin{equation}
\label{p12prim}
\begin{tabular}{l|c|c|}\cline{2-3}
{\small first row}$ \  \  \quad \rightarrow$&$ \qquad \psi_1 \qquad $ & $ \qquad \psi_3 \qquad $ \\
\cline{2-3}
{\small second row}$ \ \rightarrow$&$ \qquad \psi_2 \qquad $\\
\cline{2-2}
\end{tabular}
\end{equation}
then we can write
\begin{equation}
\label{t1}
|x_0\rangle=|0^{\otimes j}\rangle_1|1^{\otimes j}\rangle_2|x_3\rangle_3,\qquad |y\rangle=|y_1\rangle_1 |y_2\rangle_2 |y_3\rangle_3.
\end{equation}
Antisymmetric operator acts on part $\psi_1$ and $\psi_2$ of partition~\eqref{p12prim}, so $\operatorname{\sfA}=\operatorname{\sfA}_{12}$. Symmetric operator acts on parts $\psi_1,\psi_2,\psi_3$, so $\operatorname{\sfS}=\operatorname{\sfS}_{13}\operatorname{\sfS}_2$. We will use this shorthand notations.\\ As we prove in lemma~\ref{le:aux} in Appendix we have
\be
\sfS_{13}\sfS_2\sfA\sfS_2\sfS_{13}=\sfS_2^2\sfS_{13}\sfA\sfS_{13}=
\sfS_{13}\sfA\sfS_{13}\sfS_2^2,
\ee
hence 
\be
\sfS\sfA\sfS|x_0\>=(j!)^2\sfS_{13}\sfA\sfS_{13}|x_0\>, 
\label{eq:SAS-13}
\ee
since $\sfS_2^2|x_2\>_2=\sfS_2^2|1^{\otimes j}\>_2=j!j!$.

Our next task is explicit calculation of scalar products $\<x_0|\sfS_{13}\sfA\sfS_{13}|x_0\>$ and $\<x_0|\sfS_{13}\sfA\sfS_{13}|y\>$. We have 
\be
\label{eq:p1}
\begin{split}
j!\sfA \sfS_{13}|x_0\>=j!\sfA |x_2\>_2 \sfS_{13} |x_1\>_1|x_3\>_3=j!f(x_{13})\sfA |x_2\>_2 |x^S_{13}\>_{13},
\end{split}
\ee
where $f(x_{13})$ is the number of permutations which do not change  vector $|x_{13}\>_{13}$ and by superscript $S$ we denote symmetric states: $|x^S\>=\hat{\sfS}(|x\>)$, (see proof of lemma~\ref{lem:1}).
\be
\label{eq:p2}
\begin{split}
j!\sfA\sfS_{13}|y\>=j!\sfA|y_2\>_2\sfS_{13}|y_1\>_1|y_3\>_3=j!f(y_{13})\sfA|y_2\>_2|y_{13}^S\>_{13},
\end{split}
\ee
where $f(y_{13})$ is number of permutations which do not change vector $|y_{13}\>$. Finally scalar products are
\be
\begin{split}
\label{eq:p0}
(j!)^2\<x_0|\sfS_{13} \sfA \sfS_{13} |y\>&=C^2(j!)^2\<x_0|\sfS_{13} \sfA\sfA \sfS_{13} |y\>=C^2(j!)^2f(y_{13})\<x_2|_1\<\bar{x}_2|_2\sfA |\bar{y}_2\>_1|y_2\>_2\<x^S_3|x^S_3\>_3=\\
&=(-1)^{j-l+m}C^2(j!)^2f(y_{13})\<x_2|_1\<\bar{x}_2|_2\sfA |x_2\>_1|\bar{x}_2\>_2\<x^S_3|x^S_3\>_3,\\
(j!)^2\<x_0|\sfS_{13}\sfA\sfS_{13}|x_0\>&=C^2(j!)^2\<x_0|\sfS_{13}\sfA\sfA\sfS_{13}|x_0\>\<x^S_3|x^S_3\>_3=\\
&=C^2(j!)^2f^2(x_{13})\<x_2|_1\<\bar{x}_2|_2\sfA|x_2\>_1|\bar{x}_2\>_2\<x^S_3|x^S_3\>_3,
\end{split}
\ee
where $C$ is normalization factor and by bar we denote logic negation, e.g. $|x_0\>=|01\>$, then $|\bar{x}_0\>=|10\>$. Note that due to \eqref{eq:SAS-13} we have $\<x_0|\sfS \sfA \sfS|y\>=(j!)^2\<x_0|\sfS_{13} \sfA \sfS_{13} |y\>$ and $\langle x_0|\operatorname{\sfS}\operatorname{\sfA}\operatorname{\sfS}|x_0\rangle=(j!)^2\<x_0|\sfS_{13}\sfA\sfS_{13}|x_0\>$, so using equations~\eqref{eq:p0} we can write
\be
\label{eq:f1f2}
\begin{split}
\frac{\langle x_0|\operatorname{\sfS}\operatorname{\sfA}\operatorname{\sfS}|y\rangle }{\langle x_0|\operatorname{\sfS}\operatorname{\sfA}\operatorname{\sfS}|x_0\rangle}
=\frac{\<x_0|\sfS_{13}\sfA\sfS_{13}|y\>}{\<x_0|\sfS_{13}\sfA\sfS_{13}|x_0\>}=(-1)^{j-l+m}\frac{f(y_{13})}{f(x_{13})}.
\end{split}
\ee
The left-hand-side of~\eqref{eq:asa} is equal to the left-hand-side of~\eqref{eq:f1f2}, so our
last step of proof is to find constants $f(y_{13})$ and $f(x_{13})$. The  vector $x_{13}$ has  
$n-j$ entries, so we have $(n-j)!$ permutations, but only ${n-j \choose n-l}$ give us different effect, so $f(x_{13})=(n-j)!/{n-j \choose n-l}$. For the vector $y_{13}$ we have like before $(n-j)!$ permutations, but only ${n-j \choose m}$ give us a different effect. Using these arguments to formula~\eqref{eq:f1f2} we obtain statement of our lemma.
\end{proof}

\subsubsection{Proof of Theorem \ref{thm:eig-A-MM}}
\label{subsubsec:eig-A-MM}
We shall now show that the set of the basis vectors of the space $%
H_{l}^{(n)}$ may be endowed with a structure of a symmetric, commutative
association scheme. This fact gives a possibility to calculate the
eigenvalues of the operator $\rho _{l}^{(n)}$ using the results from the
theory of the algebraic association schemes.

\bigskip

Let us denote the set of binary basis vectors of the space $H_{1}^{(n)}$ by $%
B(H_{1}^{(n)})$. Thus $|B(H_{1}^{(n)})|=
{n \choose l}$ 
and the basis vectors in $B(H_{1}^{(n)})$ will be denoted $%
e_{i}\equiv e(i_{1},i_{2},...,i_{l})$ where the numbers $%
i_{1},i_{2},...,i_{l}\in \{1,2,...,,n\}$ are indices of $1^{\prime }s$ in
the basis vector $e_{i}=e(i_{1},i_{2},...,i_{l}).$ It means that in the set $%
\{1,2,...,,n\}-\{i_{1},i_{2},...,i_{l}\}$ there are indices of the $%
0^{\prime }s$ in $e(i_{1},i_{2},...,i_{l}).$

On the other hand, as it has been pointed out in Proposition~\ref{prop:BI2} the
element of the set $X^{J}$ (i.e the $l-$elements subsets of the set $V$)
may be denoted in a natural way by $x\{i_{1},i_{2},...,i_{l}\}\equiv
x\{i\}\in X^{J}$ where $i_{1},i_{2},...,i_{l}$ denote the elements from $%
V=\{1,2,...,n\}$ which are contained in the subset $x\{i_{1},i_{2},...,i_{l}%
\}\in X^{J}$ and $|X^{J}|={n \choose l}$. 
This gives us a natural bijection between the sets $B(H_{1}^{(n)})
$ and $X^{J}$%
\be
\varphi :X^{J}\rightarrow B(H_{1}^{(n)}),\quad \varphi
(x\{i_{1},i_{2},..,i_{l}\})=e(i_{1},i_{2},...,i_{l}).
\ee

From Theorem~\ref{thm:CCAS} we get that the pair $(B(H_{1}^{(n)}),
\{R_{i}^{H}\}_{i=0}^{l})$, where $R_{i}^{H}=\Phi (R_{i}^{J})$ $i=0,...,l
$ is a CAS with the same adjacency matrices as the Johnson CAS. The subsets $%
R_{i}^{H}$ of $B(H_{1}^{(n)})\times B(H_{1}^{(n)})$ are described in the
following

\begin{lemma}
\be
R_{k}^{H}=\Phi (R_{k}^{J})=\{(e_{i},e_{j})\in B(H_{1}^{(n)})\times
B(H_{1}^{(n)})\quad |\quad d(e_{i},e_{j})=2k\},\quad k=0,1,...,,l.
\ee
\end{lemma}

\begin{proof}
If $(e_{i},e_{j})\in R_{k}^{H}$ then for $x=\varphi ^{-1}(e_{i}),$ $%
y=\varphi ^{-1}(e_{j})\in X^{J}$ \ \ $(x,y)\in R_{k}^{J}\Leftrightarrow
|x\cap y|=l-k$,  and hence the $\ l-$element sets $x,y$ have $l-k$ common elements
i.e
\be
x=\{i_{1},i_{2},..,i_{k},z_{1},z_{2},..,z_{l-k}\},\quad
y=\{j_{1},j_{2},..,j_{k},z_{1},z_{2},..,z_{l-k}\}
\ee%
where%
\be
\{i_{1},i_{2},..,i_{k}\}\cap \{j_{1},j_{2},..,j_{k}\}=\varnothing
\ee%
and
\be
\varphi (x)=e(i_{1},i_{2},..,i_{k},z_{1},z_{2},..,z_{l-k}),\quad \varphi
(y)=e(j_{1},j_{2},..,j_{k},z_{1},z_{2},..,z_{l-k}).
\ee%
So the Hamming distance between the vectors $\varphi (x)=e_{i}$ and $\varphi
(y)=e_{j}$ in $B(H_{1}^{(n)})$ is equal to $2k.$
\end{proof}

From Definition~\ref{def:BI2} of the adjacency matrices it follows immediately that

\bigskip

\begin{corollary}
\be
\forall k=0,1,...,l\quad (A_{k}^{H})_{(e_{i},e_{j})}=%
\begin{array}{c}
1\quad if\quad d(e_{i},e_{j})=2k \\
0\quad if\quad d(e_{i},e_{j})\neq 2k%
\end{array}%
\ee
\end{corollary}

\begin{proof}
Now the proof of Theorem \ref{thm:eig-A-MM} follows directly from
\be
\rho_{l}^{(n)}=\frac{1}{2^{n}}\sum_{k=0}^{l}\mathcal{P}^{2k}A_{k}^{H}=\frac{1}{%
2^{n}}\sum_{k=0}^{l}\mathcal{P}^{2k}A_{k}^{J}.
\ee%
and from Proposition~\ref{prop:BI2}.
\end{proof}

\bigskip

\subsection{Examples}

In this section we present a few most interesting properties of matrix $\rho_l^{(n)}$ and their eigenvalues, especially spectral radius $\lambda_0$.

Using the identity
\begin{equation}
\sum_{r=0}^{k}(-1)^{k-r} 	{k\choose r}{m+r \choose r}{m\choose k}
\end{equation}%
one can prove

\begin{proposition}
\label{prop:jzero}
For the case $j=0$ the formula from Theorem \ref{thm:eig-A-MM}
gives%
\begin{equation}
\lambda _{l}^{n}(0)=\frac{1}{2^{n}}\sum_{k=0}^{l}p^{2k}
{k\choose l}{n-l \choose k}
=\frac{1}{2^{n}}\sum_{k=0}^{l}p^{2(l-k)}
{l\choose k}{n-l\choose l-k}
\end{equation}%
\end{proposition}

Another particular case is the following

\begin{proposition}
\label{prop:ljeden}
For the case $H_{1}^{n}$ i.e. when $l=1$ Theorem 2 gives the following
two eigenvalues%
\begin{equation}
\lambda _{1}^{n}(0)=\frac{1}{2^{n}}(1+(n-1)p^{2}),\quad \lambda
_{1}^{n}(1)=\frac{1}{2^{n}}(1-p^{2}).
\end{equation}%
\end{proposition}

\bigskip And one more particular case

\begin{proposition}
\label{prop:jjeden}
If $j=l$ then%
\begin{equation}
p_{k}(j)=(-1)^{k}
{k \choose l}
\quad \Rightarrow \quad \lambda _{l}^{n}(l)=\frac{1}{2^{n}}(1-%
p^{2})^{l}
\end{equation}%
\end{proposition}

To check our general formulas, let us obtain this latter result directly.
The matrix  $\rho _{l}^{(n)}$ has in this case the
following, so called circulant form%
\begin{equation}
\rho_{1}^{(n)}=\left(
\begin{array}{ccccc}
1 & p^2 & . & . & p^2 \\
p^2 & 1 & . & . & p^2 \\
. & . & . & . & p^2 \\
. & . & . & . & . \\
p^2 & p^2 & . & . & 1%
\end{array}%
\right) .
\end{equation}%
A direct calculation shows that the matrix $\rho_{1}^{(n)}$ has only two
distinct eigenvalues
\begin{equation}
\lambda _{0}=\frac{1}{2^{n}}(1+(n-1)p^2),\quad \lambda _{1}=%
\frac{1}{2^{n}}(1-p^2).
\end{equation}%
The first one is the spectral radius of $M_{1}^{(n)}$ with the algebraic
multiplicity $1$ and is the eigenvalue of $\rho _{l}^{(n)}$ on the subspace $%
B_{0}$ whereas the second one has multiplicity $n-1$ and is the eigenvalue
of $\rho _{l}^{(n)}$ on the subspace $B_{1}$, where in this case $%
H_{1}^{(n)}=B_{0}\oplus B_{1}.$

\section{Results}

Let us now calculate rates of the protocol for a state of Eq. \ref{MA:1} in the case when the number of copies of the state on which Alice and Bob perform the first measurement is $N=16$. In Figs. \ref{fig:1} and \ref{fig:2} we present the rates of the protocol for different values of $q$ and $x$, and for $\alpha=\frac{1}{2}$. For a given $x$ the rate is symmetric around $q=\frac{1}{2}$. It decreases from maximal value for $q=0$ to minimal value $0$ for $q=\frac{1}{2}$. For a given $q$ the rate increases from minimal value $0$ for $x=0$ to maximal value for $x=1$. Note that for $q=\frac{1}{2}$ or $x=0$ the state is separable. In Fig. \ref{fig:3} we additionaly present the rate of the protocol for different values of $\alpha$, and for $q=0.2$ and $x=0.8$. It is symmetric around $\alpha=\frac{1}{2}$. The rate increases from minimal value $0$ for $\alpha=0$, i.e., when the state is separable, to maximal value for $\alpha=\frac{1}{2}$, i.e., when the state is a mixture of two maximally entangled states and a product state.

It is also instructive to calculate rates of the protocol for a state of Eq. \ref{MA:1} for different numbers of copies of the state on which Alice and Bob perform the first measurement. In Figs. \ref{fig:4} and \ref{fig:5}  we present such rates of the protocol for  different values of $q$ and $x$, and $\alpha=\frac{1}{2}$. One can see that in general when one increases the number of copies of the state on which Alice and Bob perform the first measurement one increases the rate of the protocol. This increase of the rate is particularly observable for large $x$.

\begin{figure}
\includegraphics[width=9truecm]  {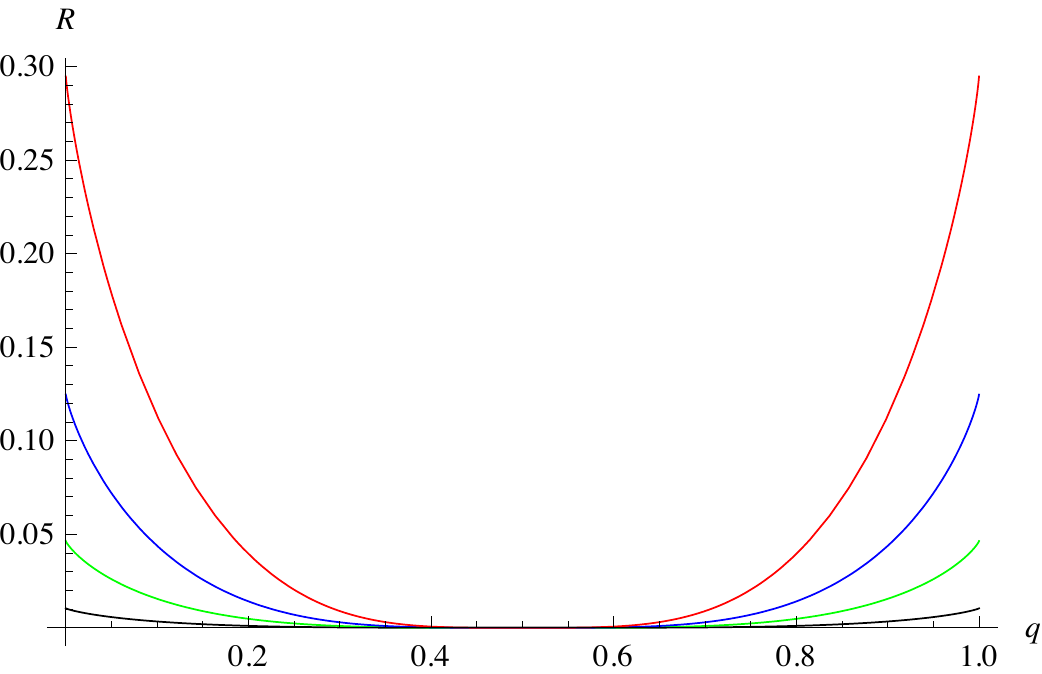}
\caption{\label{fig:1}
The rate of the protocol for different values of $x$ and $q$. From top to bottom:
x=0.8 (red), x=0.6 (blue), x=0.4 (green),  x=0.2 (black).  In all cases $N=16$ and $\alpha=\frac{1}{2}$ (color online).}
\end{figure}

\begin{figure}
\includegraphics[width=9truecm]  {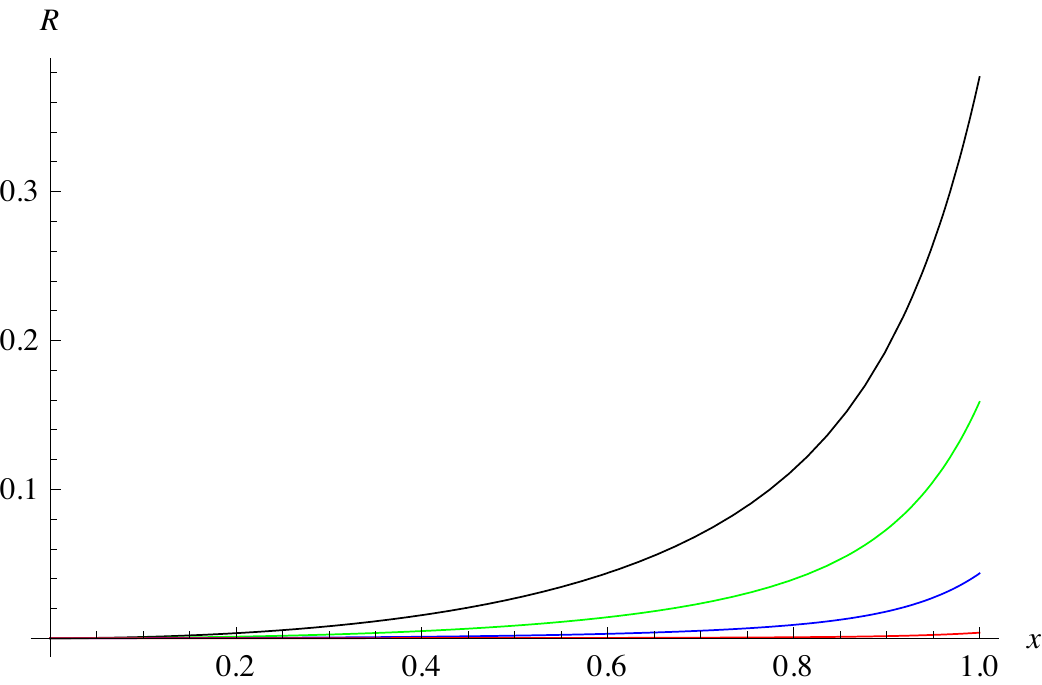}
\caption{\label{fig:2}
The rate of the protocol for different values of $x$ and $q$. From top to bottom:
q=0.1 (black), q=0.2 (green), q=0.3 (blue), q=0.4 (red).  In all cases $N=16$ and $\alpha=\frac{1}{2}$ (color online).}
\end{figure}

\begin{figure}
\includegraphics[width=9truecm]  {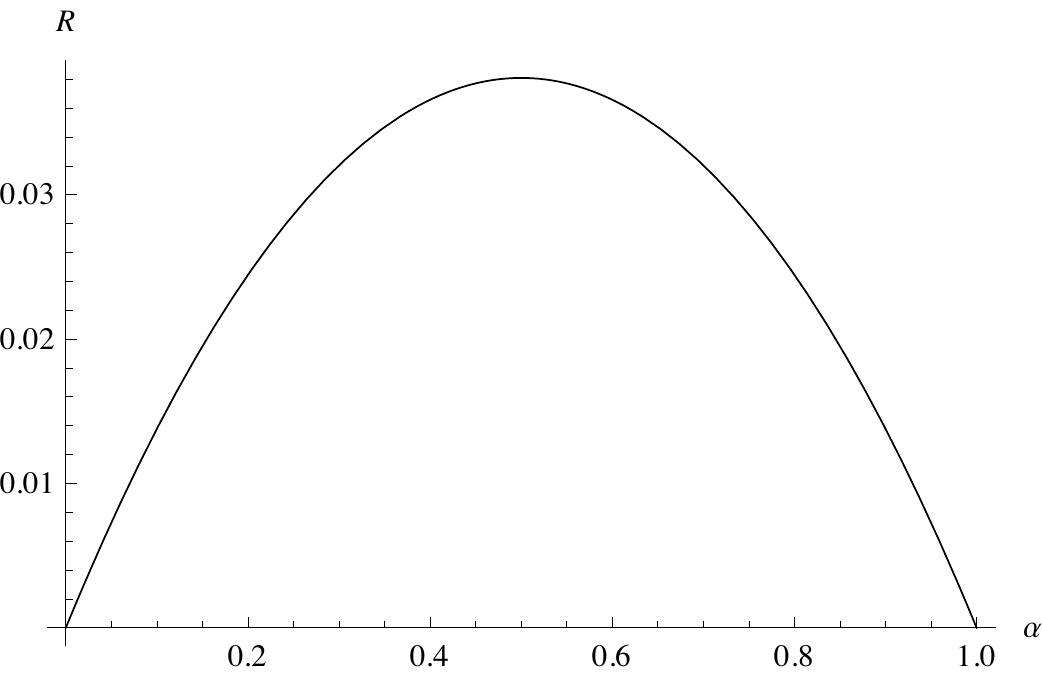}
\caption{\label{fig:3} 
The rate of the protocol for different values of $\alpha$.
N=16,  q=0.2,  x=0.8.}
\end{figure}

\begin{figure}
\includegraphics[width=9truecm]  {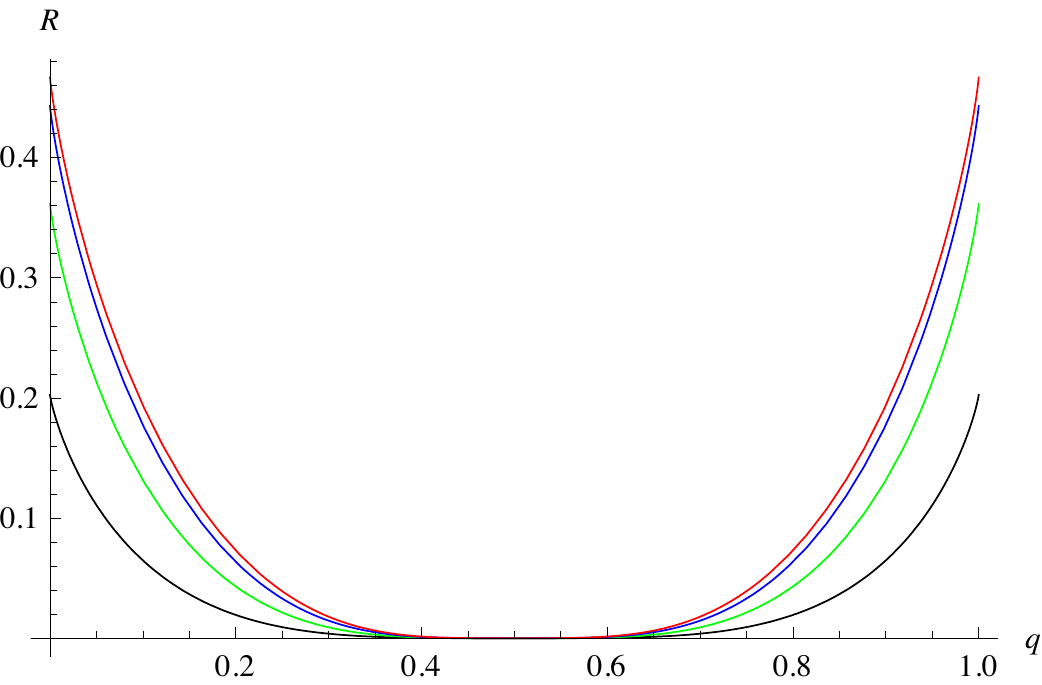}
\caption{\label{fig:4} 
The rate of the protocol for different values of $N$ and $q$. From top to bottom: N=16 (red), N=8  (blue), N=4 (green), N=2 (black). 
In all cases $x=0.9$ and $\alpha=\frac{1}{2}$ (color online).
}
\end{figure}

\begin{figure}
\includegraphics[width=9truecm]  {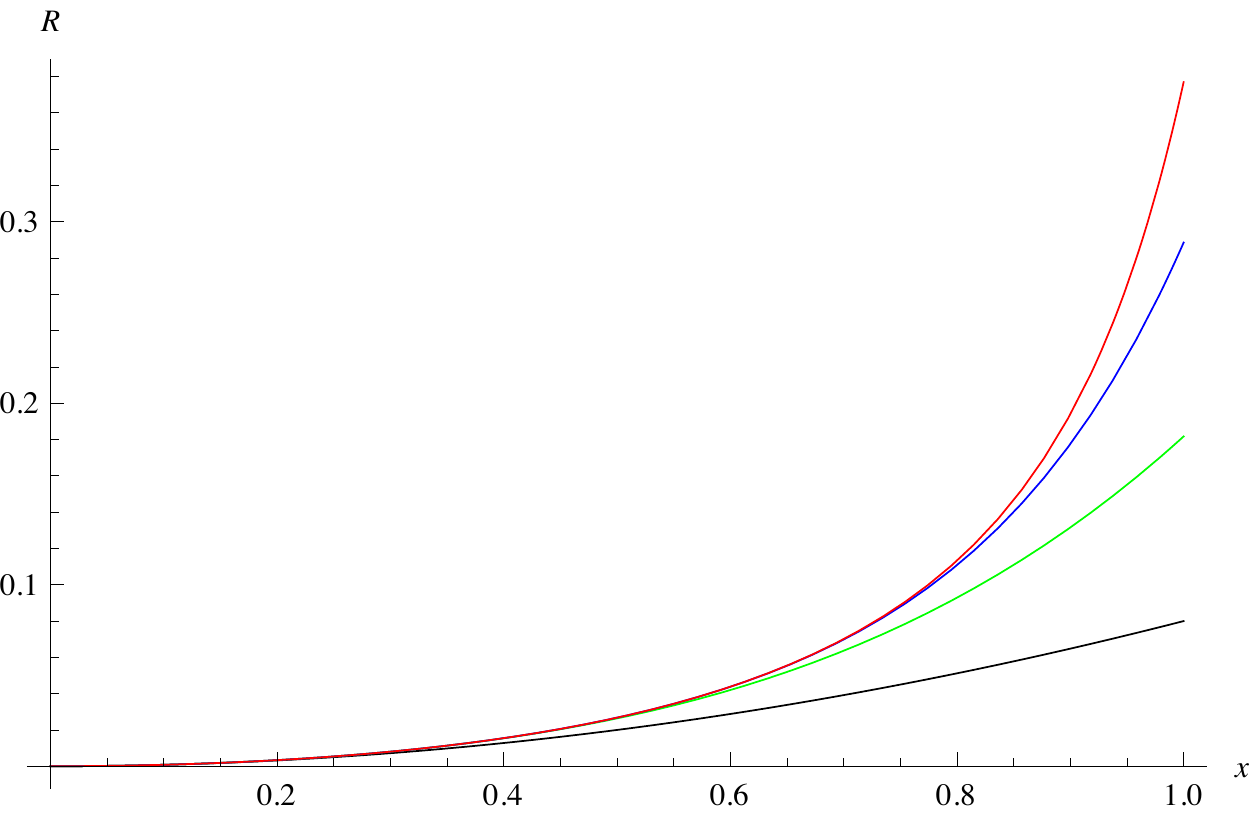}
\caption{\label{fig:5} The rate of the protocol for different values of $N$ and $x$. From top to bottom: N=16 (red), N=8  (blue), N=4 (green), N=2 (black). 
In all cases $q=0.1$ and $\alpha=\frac{1}{2}$ (color online).}
\end{figure}

\section{Appendix}
\subsubsection{The proof of decomposition of $\hcal_l^{(n})$ into irreps
by use of characters}

\begin{proposition}
Let $\chi ^{l}$ be a character of the representation $H_{l}^{(n)}$ and $%
c_{i} $ be a class $\ $\ of $S_{n}$ with cycle structure $%
(i_{1},i_{2},...,i_{n})$ (i.e. $c_{i}\equiv (i_{1},i_{2},...,i_{n})$), then%
\begin{equation}
\chi ^{l}(c_{i})=\sum_{q_{1},q_{2},..q_{n}}^{i_{1},i_{2},..i_{n}}\left(
\begin{array}{c}
i_{1} \\
q_{1}%
\end{array}%
\right) ...\left(
\begin{array}{c}
i_{n} \\
q_{n}%
\end{array}%
\right) ,
\end{equation}%
where $q_{1},q_{2},..q_{n}$ are solutions of the equation%
\begin{equation}
\sum_{k=1}^{n}kq_{k}=l
\end{equation}%
and $q_{k}\in \{0,1,...,i_{k}\}.$ \newline
\newline
\end{proposition}

\begin{proof}
Since the representation of the group $S_{n}$ on $H_{l}^{(n)}$ is a
permutation representation, the character of $\sigma \in S_{n}$ in this
representation is a number basis vectors of $H_{l}^{(n)}$ which are fixed by
$\sigma .$ So we are looking for basis vectors of $H_{l}^{(n)}$ such that
their structure, i.e. positions of $1^{\prime }s$ and $0^{\prime }s,$ makes
them invariant under the action of $\sigma \in S_{n}.$\newline
Let $\sigma \in S_{n}$ is such that $\sigma \in c_{i}\equiv
(i_{1},i_{2},...,i_{n}),$ that is $\sigma $ contains $i_{k}$ cycles of
length $k$, $k=1,2,...,,n.$\newline
For each $k,$ among all $i_{k}$ cycles of length $k$ in $c_{i}$, we choose $%
q_{k}$ cycles such that the numbers $q_{k}$ satisfy%
\begin{equation}
\sum_{k=1}^{n}kq_{k}=l.
\end{equation}%
Solutions of this equation, the numbers $q_{k},$ gives the information how
many cycles we may chose from each $i_{k}$ cycles of length $k$, in order to
get a set $L$ (which is a subset of $\{1,2,...,,n\}$) containing $l$
elements which are taken from $q_{k}$ cycles in each $i_{k}$ (for each $k)$.
It means that the set $L$ contains elements of $\ q_{k}$ cycles of length $%
k, $ for each $k.$The essential is now that the basis vector $H_{l}^{(n)}$
is invariant under action of $\sigma $ only if its $l$ $\ 1^{\prime }s$ have
indices from the set $L$ because in this situation all $l$ $\ 1^{\prime }s$
are permuted among them (and consequently the same for $0^{\prime }s)$ and
the basis vector is invariant under the action of $\sigma \in S_{n}$. The
cycle structure of this permutation of $1^{\prime }s$ is determined by
numbers $q_{k}.$\newline
Once we have the solution of the above equation i.e. the numbers $q_{k},$
then we may choose, for each $k$ separately , in $\left(
\begin{array}{c}
i_{k} \\
q_{k}%
\end{array}%
\right) $ ways the $q_{k}$ cycles of length $k$ whose elements determine the
set $L$.\newline
\newline
\end{proof}

\begin{remark}
Although this formula for the character of the representation $H_{l}^{(n)}$
of $S_{n}$ is not entirely analytic (because we do not know explicitely
solutions of the equation for the numbers $q_{k})$ it will be, together
with next Proposition, very useful in studying of representation structure
of $H_{l}^{(n)}.$\newline
\newline
\newline
\end{remark}

\begin{definition}
Let $B_{j}$ be an irreducible representation of $S_{n}$ corresponding to the
binary partition $\{n-j,j\}$ where $j\leq \frac{1}{2}n.$ Denote by $\chi
^{B_{j}}$ its character. $B_{0}$ is a trivial representation.
\end{definition}

Then we have \newline
\newline

\begin{proposition}
In the notations of the previous Proposition we have the following formula
for irreducible character $\chi ^{B_{j}}$%
\begin{equation}
\chi ^{B_{j}}(c_{i})=\sum_{q_{1},q_{2},..q_{n}}^{i_{1},i_{2},..i_{n}}\left(
\begin{array}{c}
i_{1} \\
q_{1}%
\end{array}%
\right) ...\left(
\begin{array}{c}
i_{n} \\
q_{n}%
\end{array}%
\right) -\sum_{q_{1}^{\prime },q_{2}^{\prime },..q_{n}^{\prime
}}^{i_{1},i_{2},..i_{n}}\left(
\begin{array}{c}
i_{1} \\
q_{1}^{\prime }%
\end{array}%
\right) ...\left(
\begin{array}{c}
i_{n} \\
q_{n}^{\prime }%
\end{array}%
\right)
\end{equation}%
where $q_{1},q_{2},..q_{n}$ and $q_{1}^{\prime },q_{2}^{\prime
},..q_{n}^{\prime }$ are solutions of the equations
\begin{equation}
\sum_{k=1}^{n}kq_{k}=j,\quad \sum_{k=1}^{n}kq_{k}^{\prime }=j-1
\end{equation}%
and $q_{k},q_{k}^{\prime }\in \{0,1,...,i_{k}\}.$\newline
The dimension of the representation $B_{j}$ is
\begin{equation}
\dim B_{j}=
{n \choose j}- {n\choose j-1},\quad j\geq 1;\quad \dim B_{0}=1
\end{equation}
\end{proposition}

\begin{proof}
We will calculate the character of the irreducible representation $B_{j}$
using the Frobenius formula for irreducible characters. In case of the
binary partition $\{n-j,j\}$ Frobenius formula takes the form~\cite{Fulton}%
\begin{equation}
\chi ^{B_{j}}(c_{i})\equiv \chi
^{\{n-j,j%
\}}(c_{i})=[(x_{1}-x_{2})(x_{1}+x_{2})^{i_{1}}(x_{1}^{2}+x_{2}^{2})^{i_{2}}...(x_{1}^{n}+x_{2}^{n})^{i_{n}}]_{(n-j+1,j)}
\end{equation}%
where in the parenthesis on RHS there is a polynomial in two variables $%
P(x_{1},x_{2})$ and the subscript $(n-j+1,j)$ means that the value of the
character $\chi ^{B_{j}}$ on the class $c_{i}\equiv (i_{1},i_{2},...,i_{n})$
is equal the coefficient of $x_{1}^{n-j+1}x_{2}^{j}$ in $P(x_{1},x_{2}).$
\newline
From
\begin{equation}
(x_{1}^{m}+x_{2}^{m})^{i_{m}}=\sum_{q_{m}=0}^{i_{m}}\left(
\begin{array}{c}
i_{m} \\
q_{m}%
\end{array}%
\right) x_{1}^{m(i_{m}-q_{m})}x_{2}^{mq_{m}}
\end{equation}%
we get%
\begin{equation}
(x_{1}+x_{2})^{i_{1}}(x_{1}^{2}+x_{2}^{2})^{i_{2}}...(x_{1}^{n}+x_{2}^{n})^{i_{n}}=\sum_{q_{1},q_{2},..q_{n}}^{i_{1},i_{2},..i_{n}}\left(
\begin{array}{c}
i_{1} \\
q_{1}%
\end{array}%
\right) ...\left(
\begin{array}{c}
i_{n} \\
q_{n}%
\end{array}%
\right) x_{1}^{(n-\sum_{k=1}^{n}kq_{k})}x_{2}^{(\sum_{k=1}^{n}kq_{k})}
\end{equation}%
multiplying both sides of this equation by $(x_{1}-x_{2})$ we get%
\begin{equation*}
(x_{1}-x_{2})(x_{1}+x_{2})^{i_{1}}(x_{1}^{2}+x_{2}^{2})^{i_{2}}...(x_{1}^{n}+x_{2}^{n})^{i_{n}}=
\end{equation*}
\begin{equation}
=\sum_{q_{1},q_{2},..q_{n}}^{i_{1},i_{2},..i_{n}}\left(
\begin{array}{c}
i_{1} \\
q_{1}%
\end{array}%
\right) ..\left(
\begin{array}{c}
i_{n} \\
q_{n}%
\end{array}%
\right) x_{1}^{(n+1-\sum_{k=1}^{n}kq_{k})}x_{2}^{\sum_{k=1}^{n}kq_{k})}-
\end{equation}%
\[
-\sum_{q_{1}^{\prime },q_{2}^{\prime },..q_{n}^{\prime
}}^{i_{1},i_{2},..i_{n}}\left(
\begin{array}{c}
i_{1} \\
q_{1}^{\prime }%
\end{array}%
\right) ..\left(
\begin{array}{c}
i_{n} \\
q_{n}^{\prime }%
\end{array}%
\right) x_{1}^{(n-\sum_{k=1}^{n}kq_{k}^{\prime
})}x_{2}^{(\sum_{k=1}^{n}kq_{k}^{\prime }+1)}
\]%
In order to determine the coefficient of $x_{1}^{n-j+1}x_{2}^{j}$ in RHS of
this equation we have to impose the following conditions on the powers of $%
x_{1}$ and $x_{2}$ in each sum on RHS independently
\begin{equation}
n-j+1=n+1-\sum_{k=1}^{n}kq_{k},\quad j=\sum_{k=1}^{n}kq_{k}
\end{equation}%
\begin{equation}
n-j+1=n-\sum_{k=1}^{n}kq_{k}^{\prime },\quad j=\sum_{k=1}^{n}kq_{k}^{\prime
}+1
\end{equation}%
Each pair of these equations is in fact one equation, so finally we get
following equations for numbers $q_{1},q_{2},..q_{n}$ and $q_{1}^{\prime
},q_{2}^{\prime },..q_{n}^{\prime }$
\begin{equation}
j=\sum_{k=1}^{n}kq_{k},\quad j-1=\sum_{k=1}^{n}kq_{k}^{\prime }.
\end{equation}%
which determine the coefficient of $x_{1}^{n-j+1}x_{2}^{j}$ in $%
P(x_{1},x_{2}).$
\end{proof}

\begin{remark}
Similarly as in case of character $\chi ^{l}$ of the representation $%
H_{l}^{(n)}$ it is not easy to calculate the values of the character $\chi
^{B_{j}}$ in general case, however for small values of $j$ this formula may
be useful. In fact we have
\end{remark}

\bigskip

\begin{example}
\begin{equation}
\chi ^{B_{1}}(c_{i})=i_{1}-1,\quad \chi ^{B_{2}}(c_{i})=\frac{1}{2}%
i_{1}(i_{1}-3)+i_{2},\quad \chi ^{B_{3}}(c_{i})=\frac{1}{6}%
i_{1}(i_{1}-1)(i_{1}-5)+i_{2}(i_{1}-1)+i_{3}.
\end{equation}
\end{example}

As a corollary from the above two Propositions we get a theorem describing
the structure of the representation $H_{l}^{(n)}.$ \newline
\newline

\begin{theorem}
We have the following decomposition of the representation $H_{l}^{(n)}$
\begin{equation}
H_{l}^{(n)}=\oplus _{j=0}^{l}B_{j}.
\end{equation}%
\newline
\newline
\end{theorem}

\begin{proof}
From the formulae for the characters of the representations $H_{l}^{(n)}$
and $B_{j}$ derived in previous Propositions it follows directly that
\begin{equation}
\chi ^{l}=\sum_{j=0}^{l}\chi ^{B_{j}},
\end{equation}%
and irreducible characters form a basis in the space of complex class
functions on $S_{n}$ so this decomposition of $\chi ^{l}$ is unique and it
implies the thesis of the theorem.
\end{proof}

\begin{corollary}
From this theorem it follows that for example that%
\begin{equation}
(%
\mathbb{C}
^{2})^{\otimes n}=\oplus _{j=0}^{n}(n-2j+1)B_{j},\quad
H_{l}^{(n)}=H_{l-1}^{(n)}\oplus B_{l}
\end{equation}%
and that each subspace $B_{j}$ in $H_{l}^{(n)}$ is an eigenvalue space of $%
\rho _{l}^{(n)}$. \newline
\end{corollary}

\subsubsection{Proof of the auxiliary lemmas}
\begin{lemma}
\label{le:aux}
We have following property
\be
\sfS_{13}\sfS_2\sfA\sfS_2\sfS_{13}=\sfS_2^2\sfS_{13}\sfA\sfS_{13}=
\sfS_{13}\sfA\sfS_{13}\sfS_2^2.
\ee
\end{lemma}

\begin{proof}
We will prove first equality (the second follows analogously).\\ Let $\sfX=\sum_{\pi}V_{\pi}^{(1)}\otimes V_{\pi}^{(2)}$, then we have
\be
\sfS_2\sfS_{13}=\sfX\sfS_{13}=\sfS_{13}\sfX.
\ee
Which follow from the fact $V_{\pi}^{(2)}\sfS_2=\sfS_2$ and $\sfS_2\sfS_{13}=\sfS_{13}\sfS_{2}$. Moreover $\sfX\sfA=\sfA\sfX$. Putting these properties together, we get
\be
\sfS_{13}\sfS_2\sfA\sfS_2\sfS_{13}=\sfS_{13}\sfS_2\sfA\sfX\sfS_{13}=
\sfS_2\sfS_{13}\sfX\sfA\sfS_{13}=\sfS_2^2\sfS_{13}\sfA\sfS_{13}.
\ee
\end{proof}

\begin{lemma}
\label{lem:dist}
For a basis vector $e_{i}$ in $H_{l}^{(n)}$ the number of basis vectors
whose Hamming distance to $e_{i}$ is equal $2k$ is equal to
\be
\left(
\begin{array}{c}
n-l \\
k%
\end{array}%
\right) \left(
\begin{array}{c}
l \\
k%
\end{array}%
\right) ,\quad k=0,1,...,l
\ee%
and it does not depend on the basis vector $e_{i}$.
\end{lemma}

\begin{proof}
A basis vector, is at the Hamming distance $2k$ to the vector $e_{i}$ only
if it has $l-k$ $1^{\prime }s$ in common with $e_{i}$ (it means that these $%
l-k$ \ $1^{\prime }s$ have the same position in both vectors) while its
remaining $k$ \ $1^{\prime }s$ \ are in the positions where in $e_{i}$ are $%
0^{\prime }s$. Such common $l-k$ $1^{\prime }s$ of may be chosen in 
${l \choose l-k}$ ways, whereas its remaining $k$ \ $1^{\prime }s$ may be chosen in 
${n-l \choose k}$ ways where $n-l$ is the number of $0^{\prime }s$ and these choices
are independent.
\end{proof}

\section{Acknowledgment}
M. H. would like to thank Aram Harrow for discussion. 
M. S. is supported by the International PhD Project "Physics of future quantum-based information technologies": grant MPD/2009-3/4 from Foundation for Polish Science.
A. G., M. H. and M.S. are supported by the Polish Ministry of Science 
and Higher Education grant N N202 231937. M.H. is also supported by EC IP Q-ESSENCE. 
Part of this work was done in National Quantum Information Centre of Gda\'nsk.


\begin{thebibliography}{20}
\expandafter\ifx\csname natexlab\endcsname\relax\def\natexlab#1{#1}\fi
\expandafter\ifx\csname bibnamefont\endcsname\relax
  \def\bibnamefont#1{#1}\fi
\expandafter\ifx\csname bibfnamefont\endcsname\relax
  \def\bibfnamefont#1{#1}\fi
\expandafter\ifx\csname citenamefont\endcsname\relax
  \def\citenamefont#1{#1}\fi
\expandafter\ifx\csname url\endcsname\relax
  \def\url#1{\texttt{#1}}\fi
\expandafter\ifx\csname urlprefix\endcsname\relax\def\urlprefix{URL }\fi
\providecommand{\bibinfo}[2]{#2}
\providecommand{\eprint}[2][]{\url{#2}}

\bibitem[{\citenamefont{Bennett et~al.}(1993)\citenamefont{Bennett, Brassard,
  Crepeau, Jozsa, Peres, and Wootters}}]{Bennett5}
\bibinfo{author}{\bibfnamefont{C.~H.} \bibnamefont{Bennett}},
  \bibinfo{author}{\bibfnamefont{G.}~\bibnamefont{Brassard}},
  \bibinfo{author}{\bibfnamefont{C.}~\bibnamefont{Crepeau}},
  \bibinfo{author}{\bibfnamefont{R.}~\bibnamefont{Jozsa}},
  \bibinfo{author}{\bibfnamefont{A.}~\bibnamefont{Peres}}, \bibnamefont{and}
  \bibinfo{author}{\bibfnamefont{W.~K.} \bibnamefont{Wootters}},
  \bibinfo{journal}{Phys. Rev. Lett.} \textbf{\bibinfo{volume}{70}},
  \bibinfo{pages}{1895} (\bibinfo{year}{1993}).

\bibitem[{\citenamefont{Bennett and Wiesner}(1992)}]{Bennett6}
\bibinfo{author}{\bibfnamefont{C.~H.} \bibnamefont{Bennett}} \bibnamefont{and}
  \bibinfo{author}{\bibfnamefont{S.~J.} \bibnamefont{Wiesner}},
  \bibinfo{journal}{Phys. Rev. Lett.} \textbf{\bibinfo{volume}{69}},
  \bibinfo{pages}{2881} (\bibinfo{year}{1992}).

\bibitem[{\citenamefont{Ekert}(1991)}]{Ekert1}
\bibinfo{author}{\bibfnamefont{A.~K.} \bibnamefont{Ekert}},
  \bibinfo{journal}{Phys. Rev. Lett.} \textbf{\bibinfo{volume}{67}},
  \bibinfo{pages}{661} (\bibinfo{year}{1991}).

\bibitem[{\citenamefont{Bennett
  et~al.}(1996{\natexlab{a}})\citenamefont{Bennett, Brassard, Popescu,
  Schumacher, Smolin, and Wootters}}]{Bennett1}
\bibinfo{author}{\bibfnamefont{C.~H.} \bibnamefont{Bennett}},
  \bibinfo{author}{\bibfnamefont{G.}~\bibnamefont{Brassard}},
  \bibinfo{author}{\bibfnamefont{S.}~\bibnamefont{Popescu}},
  \bibinfo{author}{\bibfnamefont{B.}~\bibnamefont{Schumacher}},
  \bibinfo{author}{\bibfnamefont{J.~A.} \bibnamefont{Smolin}},
  \bibnamefont{and} \bibinfo{author}{\bibfnamefont{W.~K.}
  \bibnamefont{Wootters}}, \bibinfo{journal}{Phys. Rev. Lett.}
  \textbf{\bibinfo{volume}{76}}, \bibinfo{pages}{722}
  (\bibinfo{year}{1996}{\natexlab{a}}).

\bibitem[{\citenamefont{Bennett
  et~al.}(1996{\natexlab{b}})\citenamefont{Bennett, DiVincenzo, Smolin, and
  Wootters}}]{Bennett3}
\bibinfo{author}{\bibfnamefont{C.~H.} \bibnamefont{Bennett}},
  \bibinfo{author}{\bibfnamefont{D.~P.} \bibnamefont{DiVincenzo}},
  \bibinfo{author}{\bibfnamefont{J.~A.} \bibnamefont{Smolin}},
  \bibnamefont{and} \bibinfo{author}{\bibfnamefont{W.~K.}
  \bibnamefont{Wootters}}, \bibinfo{journal}{Phys. Rev. A}
  \textbf{\bibinfo{volume}{54}}, \bibinfo{pages}{3824}
  (\bibinfo{year}{1996}{\natexlab{b}}).

\bibitem[{\citenamefont{Deutsch et~al.}(1996)\citenamefont{Deutsch, Ekert,
  Jozsa, Macchiavello, Popescu, and Sanpera}}]{PhysRevLett.77.2818}
\bibinfo{author}{\bibfnamefont{D.}~\bibnamefont{Deutsch}},
  \bibinfo{author}{\bibfnamefont{A.}~\bibnamefont{Ekert}},
  \bibinfo{author}{\bibfnamefont{R.}~\bibnamefont{Jozsa}},
  \bibinfo{author}{\bibfnamefont{C.}~\bibnamefont{Macchiavello}},
  \bibinfo{author}{\bibfnamefont{S.}~\bibnamefont{Popescu}}, \bibnamefont{and}
  \bibinfo{author}{\bibfnamefont{A.}~\bibnamefont{Sanpera}},
  \bibinfo{journal}{Phys. Rev. Lett.} \textbf{\bibinfo{volume}{77}},
  \bibinfo{pages}{2818} (\bibinfo{year}{1996}).

\bibitem[{\citenamefont{D{\"u}r and Briegel}(2007)}]{Dur}
\bibinfo{author}{\bibfnamefont{W.}~\bibnamefont{D{\"u}r}} \bibnamefont{and}
  \bibinfo{author}{\bibfnamefont{H.~J.} \bibnamefont{Briegel}},
  \bibinfo{journal}{Rep. Prog. Phys.} \textbf{\bibinfo{volume}{70}},
  \bibinfo{pages}{1381} (\bibinfo{year}{2007}).

\bibitem[{\citenamefont{Horodecki et~al.}(1998)\citenamefont{Horodecki,
  Horodecki, and Horodecki}}]{Horodecki7}
\bibinfo{author}{\bibfnamefont{M.}~\bibnamefont{Horodecki}},
  \bibinfo{author}{\bibfnamefont{P.}~\bibnamefont{Horodecki}},
  \bibnamefont{and}
  \bibinfo{author}{\bibfnamefont{R.}~\bibnamefont{Horodecki}},
  \bibinfo{journal}{Phys. Rev. Lett.} \textbf{\bibinfo{volume}{80}},
  \bibinfo{pages}{5239} (\bibinfo{year}{1998}).

\bibitem[{\citenamefont{Rains}(1997)}]{Rains1}
\bibinfo{author}{\bibfnamefont{E.~M.} \bibnamefont{Rains}},
  \bibinfo{journal}{arXiv:quant-ph/9707002}  (\bibinfo{year}{1997}).

\bibitem[{\citenamefont{Eisert et~al.}(2000)\citenamefont{Eisert, Felbinger,
  Papadopoulos, Plenio, and Wilkens}}]{Eisert1}
\bibinfo{author}{\bibfnamefont{J.}~\bibnamefont{Eisert}},
  \bibinfo{author}{\bibfnamefont{T.}~\bibnamefont{Felbinger}},
  \bibinfo{author}{\bibfnamefont{P.}~\bibnamefont{Papadopoulos}},
  \bibinfo{author}{\bibfnamefont{M.~B.} \bibnamefont{Plenio}},
  \bibnamefont{and} \bibinfo{author}{\bibfnamefont{M.}~\bibnamefont{Wilkens}},
  \bibinfo{journal}{Phys. Rev. Lett.} \textbf{\bibinfo{volume}{84}},
  \bibinfo{pages}{1611} (\bibinfo{year}{2000}).

\bibitem[{\citenamefont{Chen and Yang}(2002)}]{Chen1}
\bibinfo{author}{\bibfnamefont{Y.-X.} \bibnamefont{Chen}} \bibnamefont{and}
  \bibinfo{author}{\bibfnamefont{D.}~\bibnamefont{Yang}},
  \bibinfo{journal}{arXiv:quant-ph/0204004v3}  (\bibinfo{year}{2002}).

\bibitem[{\citenamefont{Hamieh and Zaraket}(2003)}]{Hamieh1}
\bibinfo{author}{\bibfnamefont{S.}~\bibnamefont{Hamieh}} \bibnamefont{and}
  \bibinfo{author}{\bibfnamefont{H.}~\bibnamefont{Zaraket}},
  \bibinfo{journal}{J. Phys. A: Math. Gen.} \textbf{\bibinfo{volume}{36}},
  \bibinfo{pages}{L387} (\bibinfo{year}{2003}).

\bibitem[{\citenamefont{Hiroshima and Hayashi}(2004)}]{Hiroshima1}
\bibinfo{author}{\bibfnamefont{T.}~\bibnamefont{Hiroshima}} \bibnamefont{and}
  \bibinfo{author}{\bibfnamefont{M.}~\bibnamefont{Hayashi}},
  \bibinfo{journal}{Phys. Rev. A} \textbf{\bibinfo{volume}{70}},
  \bibinfo{pages}{030302} (\bibinfo{year}{2004}).

\bibitem[{\citenamefont{Czechlewski et~al.}(2009)\citenamefont{Czechlewski,
  Grudka, Ishizaka, and W\'ojcik}}]{PhysRevA.80.014303}
\bibinfo{author}{\bibfnamefont{M.}~\bibnamefont{Czechlewski}},
  \bibinfo{author}{\bibfnamefont{A.}~\bibnamefont{Grudka}},
  \bibinfo{author}{\bibfnamefont{S.}~\bibnamefont{Ishizaka}}, \bibnamefont{and}
  \bibinfo{author}{\bibfnamefont{A.}~\bibnamefont{W\'ojcik}},
  \bibinfo{journal}{Phys. Rev. A} \textbf{\bibinfo{volume}{80}},
  \bibinfo{pages}{014303} (\bibinfo{year}{2009}).

\bibitem[{\citenamefont{Bannai and Ito}(1984)}]{Bannai}
\bibinfo{author}{\bibfnamefont{E.}~\bibnamefont{Bannai}} \bibnamefont{and}
  \bibinfo{author}{\bibfnamefont{T.}~\bibnamefont{Ito}},
  \emph{\bibinfo{title}{Algebraic combinatorics I}}
  (\bibinfo{publisher}{Benjamin/Cummings Publishing Company},
  \bibinfo{year}{1984}).

\bibitem[{\citenamefont{Devetak and Winter}(2004)}]{Devetak1}
\bibinfo{author}{\bibfnamefont{I.}~\bibnamefont{Devetak}} \bibnamefont{and}
  \bibinfo{author}{\bibfnamefont{A.}~\bibnamefont{Winter}},
  \bibinfo{journal}{Phys. Rev. Lett.} \textbf{\bibinfo{volume}{93}},
  \bibinfo{pages}{080501} (\bibinfo{year}{2004}).

\bibitem[{\citenamefont{Devetak and Winter}(2005)}]{Devetak2}
\bibinfo{author}{\bibfnamefont{I.}~\bibnamefont{Devetak}} \bibnamefont{and}
  \bibinfo{author}{\bibfnamefont{A.}~\bibnamefont{Winter}},
  \bibinfo{journal}{Proc. R. Soc. Lond. A} \textbf{\bibinfo{volume}{461}},
  \bibinfo{pages}{207} (\bibinfo{year}{2005}).

\bibitem[{\citenamefont{Kostrykin}(2008)}]{Kostrikin}
\bibinfo{author}{\bibfnamefont{A.~I.} \bibnamefont{Kostrykin}},
  \emph{\bibinfo{title}{Wst\c{e}p do algebry}} (\bibinfo{publisher}{Wydawnictwo
  Naukowe PWN, Warszawa}, \bibinfo{year}{2008}).

\bibitem[{\citenamefont{Audenaert}()}]{Audenaert2006-notes}
\bibinfo{author}{\bibfnamefont{K.~M.~R.} \bibnamefont{Audenaert}},
  \emph{\bibinfo{title}{A digest on representation theory of the symmetric
  group}},
  \urlprefix\url{http://www.personal.rhul.ac.uk/usah/080/QITNotes_files/Irreps_v06.pdf}.

\bibitem[{\citenamefont{Fulton and Harris}(1991)}]{Fulton}
\bibinfo{author}{\bibfnamefont{W.}~\bibnamefont{Fulton}} \bibnamefont{and}
  \bibinfo{author}{\bibfnamefont{J.}~\bibnamefont{Harris}},
  \emph{\bibinfo{title}{Representation Theory - A First Course}}
  (\bibinfo{publisher}{Springer-Verlag, New York}, \bibinfo{year}{1991}).

\end{thebibliography}
\end{document}